\documentclass[11pt]{article}

\usepackage[letterpaper,margin=1in]{geometry}
\usepackage{amsmath, amssymb, amsthm, amsfonts}
\usepackage{bbm}
\usepackage{bm}

\usepackage{subcaption}

\usepackage{ifthen}
\usepackage{comment}
\usepackage{tikz}
\usetikzlibrary{positioning,decorations.pathreplacing}
\usepackage{comment}

\usepackage{cite}
\usepackage{appendix}
\usepackage{graphicx}
\usepackage{color}
\usepackage{algorithm}
\usepackage[noend]{algorithmic}
\usepackage{epstopdf}
\usepackage{wrapfig}
\usepackage{paralist}
\usepackage[textsize=tiny]{todonotes}

\usepackage{framed}
\usepackage[framemethod=tikz]{mdframed}
\usepackage[bottom]{footmisc}
\usepackage{enumitem}
\setitemize{noitemsep,topsep=3pt,parsep=3pt,partopsep=3pt}
\usepackage[font=small]{caption}
\usepackage{xspace}

\newtheorem{theorem}{Theorem}[section]
\newtheorem{lemma}[theorem]{Lemma}
\newtheorem{meta-theorem}[theorem]{Meta-Theorem}
\newtheorem{claim}[theorem]{Claim}

\newtheorem{definition}[theorem]{Definition}

\definecolor{darkgreen}{rgb}{0,0.5,0}
\usepackage{hyperref}
\hypersetup{
    unicode=false,          % non-Latin characters in Acrobat’s bookmarks
    colorlinks=true,        % false: boxed links; true: colored links
    linkcolor=red,          % color of internal links (change box color with linkbordercolor)
    citecolor=darkgreen,        % color of links to bibliography
    filecolor=magenta,      % color of file links
    urlcolor=cyan           % color of external links
}
\usepackage[capitalize]{cleveref}
\newcommand{\eps}{\varepsilon}

\newcommand{\important}[1]{\textbf{#1}}

\newcommand{\exclude}[1]{}

\newcommand{\STOConly}[1]{}
\newcommand{\noSTOC}[1]{#1}

\newcommand{\FullOrShort}{full}

\ifthenelse{\equal{\FullOrShort}{full}}{

  \newcommand{\fullOnly}[1]{#1}
	\newcommand{\tempfullOnly}[1]{#1}
  \newcommand{\shortOnly}[1]{}
  
  }{

    \newcommand{\fullOnly}[1]{}
		\newcommand{\tempfullOnly}[1]{{\color{red}#1}}
		\newcommand{\shortOnly}[1]{#1}
    \newcommand{\IncludePictures}[1]{}
   
  }

\begin{document}

\date{}

\title{Synchronization Strings: Codes for Insertions and Deletions Approaching the Singleton Bound\footnote{Supported in part by the National Science Foundation through grants CCF-1527110 and CCF-1618280.}.}

\author{Bernhard Haeupler\\Carnegie Mellon University\\ \texttt{haeupler@cs.cmu.edu} \and
Amirbehshad Shahrasbi\\Carnegie Mellon University\\ \texttt{shahrasbi@cs.cmu.edu}}

\maketitle

\begin{abstract}
We introduce \emph{synchronization strings}, which provide \important{a novel way of efficiently dealing with \emph{synchronization errors}}, i.e., insertions and deletions. Synchronization errors are strictly more general and much harder to deal with than more commonly considered \emph{half-errors}, i.e., symbol corruptions and erasures. For every $\eps >0$, synchronization strings allow to index a sequence with an $\eps^{-O(1)}$ size alphabet such that one can \important{efficiently transform $\bm k$ synchronization errors into $\bm{(1 + \eps)k}$ half-errors}. 
This powerful new technique has many applications. In this paper, we focus on designing \emph{insdel codes}, i.e., error correcting block codes (ECCs) for insertion-deletion channels.

\smallskip

 While ECCs for both half-errors and synchronization errors have been intensely studied, the later has largely resisted progress. As Mitzenmacher puts it in his 2009 survey \cite{mitzenmacher2009survey}:
``\emph{Channels with synchronization errors \ldots are simply not adequately understood by current theory. Given the near-complete knowledge we have for channels with erasures and errors ... our lack of understanding about channels with synchronization errors is truly remarkable.}'' 
Indeed, it took until 1999 for the first insdel codes with constant rate, constant distance, and constant alphabet size to be constructed and only since 2016 are there constructions of constant rate insdel codes for asymptotically large noise rates. Even in the asymptotically large or small noise regime these codes are polynomially far from the optimal rate-distance tradeoff. This makes the understanding of insdel codes up to this work equivalent to what was known for regular ECCs after Forney introduced concatenated codes in his doctoral thesis 50 years ago. 

\smallskip

A straight forward application of our synchronization \allowbreak strings based indexing method gives a simple black-box construction which \important{transforms any ECC into an equally efficient insdel code} with only a small increase in the alphabet size. This instantly transfers much of the highly developed understanding for regular ECCs into the realm of insdel codes. Most notably, for the complete noise spectrum we obtain efficient ``near-MDS'' insdel codes which get arbitrarily close to the optimal rate-distance tradeoff given by the Singleton bound. 
%In particular, for any $\delta \in (0,1)$ and $\eps >0$ we give insdel codes with a constant alphabet size which can efficiently correct a $\delta$ fraction of insertions or deletions while achieving a rate of $1 - \delta - \eps$.
In particular, for any $\delta \in (0,1)$ and $\eps >0$ we give insdel codes achieving a rate of $1 - \delta - \eps$ over a constant size alphabet that efficiently correct a $\delta$ fraction of insertions or deletions.
\end{abstract}
	
\setcounter{page}{0}
\thispagestyle{empty}

\sloppy

\newpage

\section{Introduction}

Since the fundamental works of Shannon, Hamming, and others the field of coding theory has advanced
our understanding of how to efficiently correct symbol corruptions and erasures. The practical and theoretical
impact of error correcting codes on technology and engineering as well as mathematics, theoretical computer
science, and other fields is hard to overestimate. The problem of coding for timing errors such as closely related insertion and deletion errors, however, while also studied intensely since the 60s, has largely resisted such progress and impact so far. An expert panel~\cite{golomb1963synchronization} in 1963 concluded: \emph{``There has been one glaring hole in [Shannon's] theory; viz., uncertainties in timing, which I will propose to call time noise, have not been encompassed \ldots. Our thesis here today is that the synchronization problem is not a mere engineering detail, but a fundamental communication problem as basic as detection itself!''} however as noted in a comprehensive survey~\cite{mercier2010survey} in 2010: \emph{``Unfortunately, although it has early and often been conjectured that error-correcting codes capable of correcting timing errors could improve the overall performance of communication systems, they are quite challenging to design, which partly explains why a large collection of synchronization techniques not based on coding were developed and implemented over the years.''} or as Mitzenmacher puts in his survey~\cite{mitzenmacher2009survey}: \emph{``Channels with synchronization errors, including both insertions and deletions as well as more general timing errors, are simply not adequately understood by current theory. Given the near-complete knowledge we have for channels with erasures and errors \ldots our lack of understanding about channels with synchronization errors is truly remarkable.''} We, too, believe that the current lack of good codes and general understanding of how to handle synchronization errors is the reason why systems today still spend significant resources and efforts on keeping very tight controls on synchronization while other noise is handled more efficiently using coding techniques. We are convinced that a better theoretical understanding together with practical code constructions will eventually lead to systems which naturally and more efficiently use coding techniques to address synchronization and noise issues jointly. In addition, we feel that better understanding the combinatorial structure underlying (codes for) insertions and deletions will have impact on other parts of mathematics and theoretical computer science.

\smallskip

In this paper, we introduce synchronization strings, a new combinatorial structure which allows efficient synchronization and indexing of streams under insertions and deletions. Synchronization strings and our indexing abstraction provide a powerful and novel way to deal with synchronization issues. They make progress on the issues raised above and have applications in a large variety of settings and problems. We already found applications to channel simulations, synchronization sequences~\cite{mercier2010survey}, interactive coding schemes~\cite{gelles2015coding,kol2013interactive,HaeuplerFOCS14p235,HaeuplerSTOC14p803,HaeuplerFOCS14p403,HaeuplerSODA15p1311}, edit distance tree codes~\cite{braverman2015coding}, and error correcting codes for insertion and deletions and suspect there will be many more. In this paper we focus on the last application, namely, designing efficient error correcting block codes over large alphabets for worst-case insertion-deletion channels. 

\smallskip

The knowledge on efficient error correcting block codes for insertions and deletions, also called \emph{insdel codes}, severely lacks behind what is known for codes for Hamming errors. While Levenshtein~\cite{Levenshtein65} introduced and pushed the study of such codes already in the 60s it took until 1999 for Schulman and Zuckerman~\cite{schulman1999asymptotically} to construct the first insdel codes with constant rate, constant distance, and constant alphabet size. Very recent work of Guruswami et al.~\cite{GW-random15,GL-isit16} in 2015 and 2016 gave the first  constant rate insdel codes for asymptotically large noise rates, via list decoding. These codes are however still polynomially far from optimal in their rate or decodable distance respectively.  In particular, they achieve a rate of $\Omega(\eps^5)$ for a relative distance of $1-\eps$ or a relative distance of $O(\eps^2)$ for a rate of $1 - \eps$, for asymptotically small $\eps > 0$ (see Section~\ref{sec:relatedwork} for a more detailed discussion of related work).

\smallskip

This paper essentially closes this line of work by designing efficient ``near-MDS'' insdel codes which approach the optimal rate-distance trade-off given by the Singleton bound. We prove that for any $0 \leq \delta < 1$ and any constant $\eps > 0$, there is an efficient insdel code over a constant size alphabet with block length $n$ and rate $1 - \delta - \eps$ which can be uniquely and efficiently decoded from any $\delta n$ insertions and deletions. The code construction takes polynomial time; and encoding and decoding can be done in linear and quadratic time, respectively. 
More formally, let us define the edit distance of two given strings as the minimum number of insertions and deletions required to convert one of them to the other one.

\begin{theorem}\label{thm:main}
For any $\eps>0$ and $\delta \in (0,1)$ there exists an encoding map $E: \Sigma^k \rightarrow \Sigma^n$ and a decoding map $D: \Sigma^* \rightarrow \Sigma^k$ such that if $EditDistance(E(m),x) \leq \delta n$ then $D(x) = m$. Further $\frac{k}{n} > 1 - \delta - \eps$, $|\Sigma|=f(\eps)$, and $E$ and $D$ are explicit and can be computed in linear and quadratic time in $n$. 
\end{theorem}

We obtain this code via a black-box construction which \important{transforms any ECC into an equally efficient insdel code} with only a small increase in the alphabet size. This transformation, which is a straight forward application of our new synchronization strings based indexing method, is so simple that it can be summarized in one sentence: 
\begin{center}
\bfseries
For \emph{any} efficient length $n$ ECC with alphabet bit size $\frac{\log \eps^{-1}}{\eps}$% = \Omega(\frac{1}{\eps})}$
, attaching to every codeword, symbol by symbol, a random or suitable pseudorandom string over an alphabet of bit size $\log \eps^{-1}$ results in an efficient insdel code with a rate and decodable distance that changed by at most $\eps$. 
\end{center}

Far beyond just implying Theorem~\ref{thm:main}, this allows to instantly transfer much of the highly developed understanding for regular ECCs into the realm of insdel codes.

\smallskip

Theorem~\ref{thm:main} is obtained by using the ``near-MDS'' expander codes of Guruswami and Indyk~\cite{guruswami2005linear} as a base ECC. These codes generalize the linear time codes of Spielman~\cite{spielman1995linear} and can be encoded and decoded in linear time. Our simple encoding strategy, as outlined above, introduces essentially no additional computational complexity during encoding. Our quadratic time decoding algorithm, however, is slower than the linear time decoding of the base codes from~\cite{guruswami2005linear} but still pretty fast. In particular, a quadratic time decoding for an insdel code is generally very good given that, in contrast to Hamming codes, even computing the distance between the received and the sent/decoded string is an edit distance computation. Edit distance computations in general do usually not run in sub-quadratic time, which is not surprising given the recent SETH-conditional lower bounds~\cite{backurs2015edit}. For the settings of for insertion-only and deletion-only errors we furthermore achieve analogs of Theorem~\ref{thm:main} with linear decoding complexities.

\smallskip

In terms of the dependence of the alphabet bit size on the parameter $\eps$, which characterizes how close a code is to achieving an optimal rate/distance pair summing to one, our transformation seem to inherently produce an alphabet bit size that is near linear in $\frac{1}{\eps}$. However, the same is true for the state of the art linear-time base ECCs~\cite{guruswami2005linear} which have an alphabet bit size of $\Theta(\frac{1}{\eps^2})$. Existentially it is known that an alphabet bit size logarithmic in $\frac{1}{\eps}$ is necessary and sufficient and ECCs based on algebraic geometry~\cite{tsfasman2013algebraic} achieving such a bound up to constants are known, but their encoding and decoding complexities are higher.

\subsection{High-level Overview, Intuition and Overall Organization}\label{sec:overview}

While extremely powerful, the concept and idea behind synchronization strings is easily demonstrated. In this section, we explain the high-level approach taken and provide intuition for the formal definitions and proofs to follow. This section also explains the overall organization of the rest of the paper. 

\subsubsection{Synchronization Errors and Half-Errors}\label{sec:synchandHE}

Consider a stream of symbols over a large but constant size alphabet $\Sigma$ in which some constant fraction $\delta$ of symbols is corrupted. 

There are two basic types of corruptions we will consider, half-errors and synchronization errors. 
 \emph{Half-errors} consist of \emph{erasures}, that is, a symbol being replaced with a special ``?'' symbol indicating the erasure, and \emph{symbol corruptions} in which a symbol is replaced with any other symbol in $\Sigma$. The wording half-error comes from the realization that when it comes to code distances erasures are half as bad as symbol corruptions. An erasure is thus counted as one half-error while a symbol corruption counts as two half-errors (see Section~\ref{sec:def} for more details). 
 \emph{Synchronization errors} consist of \emph{deletions}, that is, a symbol being removed without replacement, and \emph{insertions}, where a new symbol from $\Sigma$ is added anywhere. 

\smallskip
  
It is clear that \important{synchronization errors are strictly more general and harsher than half-errors}. In particular, any symbol corruption, worth two half-errors, can also be achieved via a deletion followed by an insertion. Any erasure can furthermore be interpreted as a deletion together with the often very helpful extra information where this deletion took place. This makes synchronization errors at least as hard as half-errors. The real problem that synchronization errors bring with them however is that they cause sending and receiving parties to become ``out of synch''. This easily changes how received symbols are interpreted and makes designing codes or other systems tolerant to synchronization errors an inherently difficult and significantly less well understood problem. 

\subsubsection{Indexing and Synchronization Strings: Reducing Synchronization Errors to Half-Errors}

There is a simple folklore strategy, which we call \emph{indexing}, that avoids these synchronization problems: Simply enhance any element with a time stamp or element count. More precisely, consecutively number the elements and attach this position count or \emph{index} to each stream element. Now, if we deal with only deletions it is clear that the position of any deletion is easily identified via a missing index, thus transforming it into an erasure. Insertions can be handled similarly by treating any stream index which is received more than once as erased. If both insertions and deletions are allowed one might still have elements with a spoofed or incorrectly received index position caused by a deletion of an indexed symbol which is then replaced by a different symbol with the same index. This however requires two insdel errors. Generally this \emph{trivial indexing strategy} can seen to successfully \emph{transform any $k$ synchronization errors into at most $k$ half-errors}. 

In many applications, however, this trivial indexing cannot be used, because having to attach a $\log n$ bit\footnote{Throughout this paper all logarithms are binary.} long index description to each element of an $n$ long stream is prohibitively costly. Consider for example an error correcting code of constant rate $R$ over some potentially large but nonetheless constant size alphabet $\Sigma$, which encodes $\frac{R n}{\log |\Sigma|}$ bits into $n$ symbols from $\Sigma$. Increasing $\Sigma$ by a factor of $n$ to allow each symbol to carry its $\log n$ bit index would destroy the desirable property of having an alphabet which is independent from the block length $n$ and would furthermore reduce the rate of the code from $R$ to $
\Theta(\frac{R}{\log n})$, which approaches zero for large block lengths. For streams of unknown or infinite length such problems become even more pronounced. 

This is where \emph{synchronization strings} come to the rescue. Essentially, synchronization strings allow to \important{index every element in an infinite stream using only a constant size alphabet} while achieving an arbitrarily good approximate reduction from synchronization errors to half-errors. In particular, using synchronization strings \important{$\bm k$ synchronization errors can be transformed into at most $\bm{ (1 + \eps)k}$ half-errors using an alphabet of size independent of the stream length} and in fact only polynomial in $\frac{1}{\eps}$. Moreover, these synchronization strings have simple constructions and fast and easy decoding procedures. 

Attaching our synchronization strings to the codewords of any efficient error correcting code, which efficiently tolerates the usual symbol corruptions and erasures, transforms any such code into an efficiently decodable insdel code while only requiring a negligible increasing in the alphabet size. This allows to use the decades of intense research in coding theory for Hamming-type errors to be transferred into the much harder and less well understood insertion-deletion setting.

\subsection{Synchronization Strings: Definition, Construction, and Decoding}\label{sec:introsynchstrings}

Next, we want to briefly motivate and explain how we arrive at a natural definition of these magical indexing sequences $S$ over a finite alphabet $\Sigma$ and what intuition lies behind their efficient constructions and decoding procedures.
 
Suppose a sender has attached some indexing sequence $S$ one-by-one to each element in a stream and consider a time $t$ at which a receiver has received a corrupted sequence of the first $t$ index descriptors, i.e., a corrupted version of the length $t$ prefix of $S$. When the receiver tries to guess or decode the current index it should naturally consider all indexing symbols received so far and find the ``best'' prefix of $S$. This suggests that the prefix of length $l$ of a synchronization string $S$ acts as a codeword for the index position $l$ and that one should think of the set of prefixes of $S$ as a code associated with the synchronization string $S$. Naturally one would want such a code to have good distance properties between any two codewords under some distance measure. While edit distance, i.e., the number of insertions and deletions needed to transform one string into another seems like the right notion of distance for insdel errors in general, the prefix nature of the codes under consideration will guarantee that codewords for indices $l$ and $l'>l$ will have edit distance exactly $l'-l$. This implies that even two very long codewords only have a tiny edit distance. On the one hand, this precludes synchronization codes with a large relative edit distance between its codewords. On the other hand, one should see this phenomenon as simply capturing the fact that at any time a simple insertion of an incorrect symbol carrying the correct next indexing symbol will lead to an unavoidable decoding error. Given this natural and unavoidable sensitivity of synchronization codes to recent corruptions, it makes sense to instead use a distance measure which captures the recent density of errors. In this spirit, we suggest the definition of a, to our knowledge, new string distance measure which we call \emph{relative suffix distance}, which intuitively measures the worst fraction of insdel errors to transform suffixes, i.e., recently sent parts of two strings, into each other. 
 This natural measure, in contrast to a similar measure defined in~\cite{braverman2015coding}, turns out to induce a metric space on any set of strings. 

\smallskip

With this natural definitions for an induced set of codewords and a natural distance metric associated with any such set the next task is to design a string $S$ for which the set of codewords has as large of a minimum pairwise distance as possible. When looking for (infinite) sequences that induce such a set of codewords and thus can be successfully used as synchronization strings it became apparent that one is looking for highly irregular and non-self-similar strings over a fixed alphabet $\Sigma$. It turns out that the correct definition to capture these desired properties, which we call $\eps$-synchronization property, states that any two neighboring intervals of $S$ with total length $l$ should require at least $(1-\eps) l$ insertions and deletions to transform one into the other, where $\eps \geq 0$. A one line calculation also shows that this clean  property also implies a large minimum relative suffix distance between any two codewords. Not surprisingly, random strings essentially satisfy this $\eps$-synchronization property, except for local imperfections of self-similarity, such as, symbols repeated twice in a row, which would naturally occur in random sequences about every $|\Sigma|$ positions. This allows us to use the probabilistic method and the general Lov\'{a}sz Local Lemma to prove the existence $\eps$-synchronization strings. This also leads to an efficient randomized construction. 

\smallskip

Finally, decoding any string to the closest codeword, i.e., the prefix of the synchronization string $S$ with the smallest relative suffix distance, can be easily done in polynomial time because the set of synchronization codewords is linear and not exponential in $n$ and (edit) distance computations (to each codeword individually) can be done via the classical Wagner-Fischer dynamic programming approach.

\subsection{More Sophisticated Decoding Procedures}\label{sec:introglobaldecoding}

All this provides an indexing solution which transforms any $k$ synchronization errors into at most $(5+\eps)k$ half-errors. This already leads to insdel codes which achieve a rate approaching $1 - 5\delta$ for any $\delta$ fraction of insdel errors with $\delta < \frac{1}{5}$. While this is already a drastic improvement over the previously best $1 - O(\sqrt{\delta})$ rate codes from \cite{GL-isit16}, which worked only for sufficiently small $\delta$, it is a far less strong result than the near-MDS codes we promised in Theorem~\ref{thm:main} for every $\delta \in (0,1)$.

\smallskip

We were able to improve upon the above strategy slightly by considering an alternative to the relative suffix distance measure, which we call relative suffix pseudo distance RSPD. RSPD was introduced in \cite{braverman2015coding} and while neither being symmetric nor satisfying the triangle inequality, can act as a pseudo distance in the minimum-distance decoder. For any set of $k=k_i+k_d$ insdel errors consisting of $k_i$ insertions and $k_d$ deletions this improved indexing solution leads to at most $(1 + \eps)(3k_i+k_d)$ half-errors which already implies near-MDS codes for deletion-only channels but still falls short for general insdel errors. We leave open the question whether an improved pseudo distance definition can achieve an indexing solution with negligible number of misdecodings for a minimum-distance decoder.

\smallskip 

In order to achieve our main theorem we developed an different strategy. Fortunately, it turned out that achieving a better indexing solution and the desired insdel codes does not require any changes to the definition of synchronization codes, the indexing approach itself, or the encoding scheme but solely required a very different decoding strategy. In particular, instead of decoding indices in a streaming manner we consider more global decoding algorithms. We provide several such decoding algorithms in Section~\ref{sec:improved_decoding}. In particular, we give a simple global decoding algorithm which for which the number of misdecodings goes to zero as the quality $\eps$ of the $\eps$-synchronization string used goes to zero, irrespectively of how many insdel errors are applied. 
	
\smallskip

Our global decoding algorithms crucially build on another key-property which we prove holds for any $\eps$-synchronization string $S$, namely that there is no monotone matching between $S$ and itself which mismatches more than a $\eps$ fraction of indices. Besides being used in our proofs, considering this $\eps$-self-matching property has another advantage. We show that this property is achieved easier than the full $\eps$-synchronization property and that indeed a random string satisfies it with good probability. This means that, in the context of error correcting codes, one can even use a simple uniformly random string as a ``synchronization string''. Lastly, we show that even a $n^{-O(1)}$-approximate $O\left(\frac{\log n}{\log \frac{1}{\eps}}\right)$-wise independent random strings satisfy the desired $\eps$-self-matching property which, using the celebrated small sample space constructions from~\cite{naor1993small} also leads to a deterministic polynomial time construction. 

\smallskip

Lastly, we provide simpler and faster global decoding algorithms for the setting of deletion-only and insertion-only corruptions. These algorithms are essentially greedy algorithms which run in linear time. They furthermore guarantee that their indexing decoding is error-free, i.e., they only output ``I don't know'' for some indices but never produce an incorrectly decoded index. Such decoding schemes have the advantage that one can use them in conjunction with error correcting codes that efficiently recover from erasures (and not necessarily also symbol corruptions). 
	
\subsection{Organization of this Paper}

The organization of this paper closely follows the flow of the high-level description above. 

We start by giving more details on related work in Section~\ref{sec:relatedwork} and introduce notation used in the paper in Section~\ref{sec:def} together with a formal introduction of the two different error types as well as (efficient) error correcting codes and insdel codes. In Section~\ref{sec:indexing}, we formalize the indexing problem and (approximate) solutions to it. Section~\ref{sec:codings} shows how any solution to the indexing problem can be used to transform any regular error correcting codes into an insdel code. Section~\ref{sec:sync} introduces the relative suffix distance and $\eps$-synchronization strings, proves the existence of $\eps$-synchronization strings and provides an efficient construction. Section~\ref{sec:sync_decoding} shows that the minimum suffix distance decoder is efficient and leads to a good indexing solution. We elaborate on the connection between $\eps$-synchronization strings and the $\eps$-self-matching property in Section~\ref{sec:SelfMatching} and provide our improved decoding algorithms in the remainder of Section~\ref{sec:improved_decoding}. 

\global\def\RelatedWork{
\subsection{Related Work}\label{sec:relatedwork} 

Shannon was the first to systematically study reliable
communication. He introduced random error channels, defined
information quantities, and gave probabilistic existence proofs of
good codes. Hamming was the first to look at worst-case errors and
code distances as introduced above. Simple counting arguments on the
volume of balls around codewords given in the 50's by Hamming and
Gilbert-Varshamov produce simple bounds on the rate of $q$-ary codes
with relative distance $\delta$. In particular, they show the
existence of codes with relative distance $\delta$ and rate at least
$1 - H_q(\delta)$ where $H_q(x) = x \log(q-1) - \frac{x \log x -
(1-x) \log (1-x)}{\log q}$ is the $q$-ary entropy function. This means
that for any $\delta < 1$ and $q = \omega(1/\delta)$ there exists codes with distance
$\delta$ and rate approaching $1 - \delta$. Concatenated codes and the
generalized minimum distance decoding procedure introduced by Forney
in 1966 led to the first codes which could recover from constant error fractions $\delta \in (0,1)$ while having polynomial time encoding and decoding procedures. The rate achieved by concatenated codes for large alphabets with sufficiently small distance $\delta$ comes out to be $1 - O(\sqrt{\delta})$. On the other hand, for $\delta$ sufficiently close to one, one can achieve a constant rate of $O(\delta^2)$. Algebraic geometry codes suggested by Goppa in 1975 later lead to error correcting codes which for every $\eps >0$ achieve the optimal rate of $1 - \delta - \eps$ with an alphabet size polynomial in $\eps$ while being able to efficiently correct for a $\delta$ fraction of half-errors~\cite{tsfasman2013algebraic}. 

While this answered the most basic questions, research since then has
developed a tremendously powerful toolbox and selection of explicit
codes. It attests to the importance of error correcting codes that
over the last several decades this research direction has developed into the
incredibly active field of coding theory with hundreds of researchers
studying and developing better codes. A small and highly incomplete
subset of important innovations include rateless codes, such as, LT
codes~\cite{LTcodes}, which do not require to fix a desired distance
at the time of encoding, explicit expander
codes~\cite{spielman1995linear,guruswami2005linear} which allow linear time encoding and
decoding, polar codes~\cite{guruswami2015polar,guruswami2015entropy} which
can approach Shannon's capacity polynomially fast, network
codes~\cite{li2003linear} which allow intermediate nodes in a network
to recombine codewords, and efficiently list decodable
codes~\cite{guruswami2008explicit} which allow to list-decode codes of
relative distance $\delta$ up to a fraction of about $\delta$ symbol corruptions.

While error correcting codes for insertions and deletions have also been intensely studied, our understanding of them is much less well developed. We refer to the 2002 survey by Sloan~\cite{sloane2002single} on single-deletion codes, the 2009 survey by Mitzenmacher~\cite{mitzenmacher2009survey} on codes for random deletions and the most general 2010 survey by Mercier et al.~\cite{mercier2010survey} for the extensive work done around codes for synchronization errors and only mention the results most closely related to Theorem~\ref{thm:main} here:  Insdel codes were first considered by Levenshtein~\cite{Levenshtein65} and since then many bounds and constructions for such codes have been given. However, while essentially the same volume and sphere packing arguments as for regular codes show that there exists insdel codes capable of correcting a fraction $\delta$ of insdel erros with rate $1 - \delta$, no efficient constructions anywhere close to this rate-distance tradeoff are known. Even the construction of efficient insdel codes over a constant alphabet with any (tiny) constant relative distance and any (tiny) constant rate had to wait until Schulman and Zuckerman gave the first such code in 1999~\cite{schulman1999asymptotically}. Over the last two years Guruswami et al. provided new codes improving over this state of the art the asymptotically small or large noise regime by giving the first codes which achieve a constant rate for noise rates going to one and codes which provide a rate going to one for an asymptotically small noise rate.  In particular, \cite{GW-random15} gave the first  efficient codes codes over fixed alphabets to correct a deletion  fraction approaching $1$, as well as efficient binary codes to correct a small constant fraction of deletions with rate approaching $1$. These codes could, however, only be efficiently decoded for deletions and not insertions. A follow-up work gave new and improved codes with similar rate-distance tradeoffs which can be efficiently decoded from insertions and deletions~\cite{GL-isit16}. In particular, these codes achieve a rate of $\Omega(\delta^5)$ and $1 - \tilde{O}(\sqrt{\delta})$ while being able to efficiently recover from a $\delta$ fraction of insertions and deletions. These works put the current state of the art for error correcting codes for insertions and deletions pretty much equal to what was known for regular error correcting codes 50 years ago, after Forney's 1965 doctoral thesis. 
}
\RelatedWork
%\fullOnly{\RelatedWork}

\section{Definitions and Preliminaries}\label{sec:def}
In this section, we provide the notation and definitions we will use throughout the rest of the paper. 

\subsection{String Notation and Edit Distance}

\noindent \textbf{String Notation.} For two strings $S \in \Sigma^n$ and $S' \in \Sigma^{n'}$ be two strings over alphabet $\Sigma$. We define $S \cdot S' \in \Sigma^{n+n'}$ to be their concatenation. For any positive integer $k$ we define $S^k$ to equal $k$ copies of $S$ concatenated together. For $i,j \in \{1, \dots, n\}$, we denote the substring of $S$ from the $i^{th}$ index through and including the $j^{th}$ index as $S[i,j]$. Such a consecutive substring is also called a \emph{factor} of $S$. For $i < 1$ we define $S[i,j] = \bot^{-i+1} \cdot S[1,j]$ where $\bot$ is a special symbol not contained in $\Sigma$. We refer to the substring from the $i^{th}$ index through, but not including, the $j^{th}$ index as $S[i,j)$. The substrings $S(i,j]$ and $S[i,j]$ are similarly defined. Finally, $S[i]$ denotes the $i^{th}$ symbol of $S$ and $|S| = n$ is the length of $S$. Occasionally, the alphabets we use are the cross-product of several alphabets, i.e. $\Sigma = \Sigma_1 \times \cdots \times \Sigma_n$. If $T$ is a string over $\Sigma,$ then we write $T[i] = \left[a_1, \dots, a_n\right]$, where $a_i \in \Sigma_i$. 

\smallskip

\noindent \textbf{Edit Distance.} Throughout this work, we rely on the well-known \emph{edit distance} metric defined as follows.
\begin{definition}[Edit distance]
The \emph{edit distance} $ED(c,c')$ between two strings $c,c' \in \Sigma^*$ is the minimum number of insertions and deletions required to transform $c$ into $c'$.
\end{definition}
It is easy to see that edit distance is a metric on any set of strings and in particular is symmetric and satisfies the triangle inequality property. Furthermore, $ED\left(c,c'\right) = |c| + |c'| - 2\cdot LCS\left(c,c'\right)$, where $LCS\left(c,c'\right)$ is the longest common substring of $c$ and $c'$.

\smallskip

We also use some \emph{string matching} notation from ~\cite{braverman2015coding}:

\begin{definition}[String matching]
Suppose that $c$ \allowbreak and $c'$ are two strings in $\Sigma^*$, and suppose that $*$ is a symbol not in $\Sigma$. Next, suppose that there exist two strings $\tau_1$ and $\tau_2$ in $\left(\Sigma \cup \{*\}\right)^*$ such that $|\tau_1| = |\tau_2|$, $del\left(\tau_1\right) = c$, $del(\tau_2) = c'$, and $\tau_1[i] \approx \tau_2[i]$ for all $i \in \left\{1, \dots, |\tau_1|\right\}$. Here, $del$ is a function that deletes every $*$ in the input string and $a \approx b$ if $a = b$ or one of $a$ or $b$ is $*$. Then we say that $\tau = \left(\tau_1, \tau_2\right)$ is a \emph{string matching} between $c$ and $c'$ (denoted $\tau: c \to c'$). We furthermore denote with $sc\left(\tau_i\right)$ the number of $*$'s in $\tau_i$.
\end{definition}

Note that the \emph{edit distance} $ED(c,c')$ between strings $c,c, \in \Sigma^*$ is exactly equal to $\min_{\tau: c \to c'}\left\{sc\left(\tau_1\right) + sc\left(\tau_2\right)\right\}$.

\global\def\ECCPrelims{

\subsection{Error Correcting Codes}

Next we give a quick summary of the standard definitions and formalism around error correcting codes. This is mainly for completeness and we remark that readers already familiar with basic notions of error correcting codes might want to skip this part.

\paragraph{Codes, Distance, Rate, and Half-Errors}

\newcommand{\Sigmain}{{\Sigma'}}
\newcommand{\Sigmaout}{{\Sigma}}

An \emph{error correcting code} $C$ is an injective function which
takes an input string $s \in (\Sigmain)^{n'}$ over alphabet $\Sigmain$
of length $n'$ and generates a \emph{codeword} $C(s) \in \Sigmaout^n$
of length $n$ over alphabet $\Sigmaout$. The length $n$ of a codeword
is also called the \emph{block length}. The two most important
parameters of a code are its distance $\Delta$ and its rate
$R$. The \emph{rate} $R = \frac{n \log |\Sigmaout|}{n' \log
|\Sigmain|}$ measures what fraction of bits in the codewords produced
by $C$ carries non-redundant information about the
input. The \emph{code distance} $\Delta(C)
= \min_{s,s'} \Delta(C(s),C(s'))$ is simply the minimum Hamming
distance between any two codewords. The \emph{relative distance}
$\delta(C) = \frac{\Delta(C)}{n}$ measures what fraction of output
symbols need to be corrupted to transform one codeword into another.

It is easy to see that if a sender sends out a codeword $C(s)$ of code $C$ with relative distance $\delta$ a receiver can uniquely recover $s$ if she receives a codeword in which less than a $\delta$ fraction of symbols are affected by an \emph{erasure}, i.e., replaced by a special ``$?$'' symbol.
Similarly, a receiver can uniquely recover the input $s$ if less than
$\delta / 2$ \emph{symbol corruptions}, in which a symbol is replaced
by any other symbol from $\Sigmaout$, occurred. More generally it is
easy to see that a receiver can recover from any combination of
$k_{e}$ erasures and $k_{c}$ corruptions as long as $k_{e} + 2 k_{c}
< \delta n$. This motivates defining half-errors to incorporate both
erasures and symbol corruptions where an erasure is counted as a
single half-error and a symbol corruption is counted as two
half-errors. In summary, any code of distance $\delta$ can tolerate
any error pattern of less than $\delta n$ half-errors.

We remark that in addition to studying codes with decoding guarantees
for worst-case error pattern as above one can also look at more benign
error models which assume a distribution over error patterns, such as
errors occurring independently at random. In such a setting one looks
for codes which allow unique recovery for typical error patterns,
i.e., one wants to recover the input with probability tending to $1$
rapidly as the block length $n$ grows. While synchronization strings 
might have applications for such codes as well, this paper focuses
exclusively on codes with good distance guarantees which tolerate an
arbitrary (worst-case) error pattern.

\paragraph{Synchronization Errors}
In addition to half-errors, we study \emph{synchronization errors}
which consist of \emph{deletions}, that is, a symbol being removed
without replacement, and \emph{insertions}, where a new symbol from
$\Sigmaout$ is added anywhere. It is clear
that \important{synchronization errors are strictly more general and
harsh than half-errors} (see Section~\ref{sec:synchandHE}). The above formalism of codes, rate, and distance works equally well for synchronization errors if one replaces the Hamming distance with edit distance. Instead of measuring the number of symbol corruptions required to transform one string into another, \emph{edit distance}
measures the minimum number of insertions and deletions to do so. An insertion-deletion error correcting code, or \emph{insdel code} for short, of relative distance $\delta$ is a set of codewords for which at least $\delta n$ insertions and deletions are needed to transformed any codeword into another. Such a code can correct any combination of less than $\delta n/2$ insertions and deletions. We remark that it is possible for two codewords of length $n$ to have edit distance up to $2n$ putting the (minimum) relative edit distance between zero and two and allowing for constant rate codes which can tolerate $(1 - \eps)n$ insdel errors. 

\paragraph{Efficient Codes}
In addition to codes with a good minimum distance, one furthermore
wants efficient algorithms for the encoding and error-correction tasks
associated with the code. Throughout this paper we say a code is
efficient if it has encoding and decoding algorithms running in time
polynomial in the block length. While it is often not hard to show
that random codes exhibit a good rate and distance, designing codes
which can be decoded efficiently is much harder. We remark that most
codes which can efficiently correct for symbol corruptions are also
efficient for half-errors. For insdel codes the situation is slightly
different. While it remains true that any code that can uniquely be
decoded from any $\delta(C)$ fraction of deletions can also be decoded
from the same fraction of insertions and
deletions~\cite{Levenshtein65} doing so efficiently is often much
easier for the deletion-only setting than the fully general insdel
setting. 
}

\fullOnly{\ECCPrelims}
\shortOnly{\smallskip \tempfullOnly{Basic notations and definitions on error correcting codes can be found in Appendix~\ref{app:ECCPrelims}}
\tempfullOnly{\ECCPrelims}}.

\section{The Indexing Problem}\label{sec:indexing}
In this section, we formally define the indexing problem. In a nutshell, this problem is that of sending a suitably chosen string $S$ of length $n$ over an insertion-deletion channel such that the receiver will be able to figure out the indices of most of the symbols he receives correctly. This problem can be trivially solved by sending the string $S = 1,2,\ldots, n$ over the alphabet $\Sigma = \{1,\ldots,n\}$ of size $n$. Interesting solution to the indexing problem, however, do almost as well while using a finite size alphabet. While very intuitive and simple, the formalization of this problem and its solutions enables an easy use in many applications. 

To set up an $(n,\delta)$-indexing problem, we fix $n$, i.e., the number of symbols which are being sent, and the maximum fraction $\delta$ of symbols that can be inserted or deleted. We further call the string $S$ the \emph{synchronization string}. Lastly, we describe the influences of the $n\delta$ worst-case insertions and deletions which transform $S$ into the related string $S_\tau$ in terms of a string matching $\tau$. In particular, $\tau=(\tau_1, \tau_2)$ is the string matching from $S$ to $S_\tau$ such that $del(\tau_1) = S$, $del(\tau_2) = S_\tau$, and for every $k$
$$(\tau_1[k], \tau_2[k]) = \Bigg\{
\begin{tabular}{ll}
$(S[i], *)$ & if $S[i]$ is deleted\\
$(S[i], S_\tau[j])$ & if $S[i]$ is delivered as $S_\tau[j]$\\
$(*, S_\tau[j])$ & if $S_\tau[j]$ is inserted
\end{tabular}$$
where $i = |del(\tau_1[1, k])|$ and $j=|del(\tau_2[1, k])|$.

\begin{definition}[$(n,\delta)$-Indexing Algorithm]
The pair $(S, \mathcal{D}_S)$ consisting of a synchronization string $S\in \Sigma^n$ and an algorithm $\mathcal{D}_S$ is called a $(n,\delta)$-indexing algorithm over alphabet $\Sigma$ if for any set of $n \delta$ insertions and deletions represented by $\tau$ which alter $S$ to a string $S_\tau$, the algorithm $\mathcal{D}_S(S_\tau)$ outputs either $\bot$ or an index between $1$ and $n$ for every symbol in $S_{\tau}$. 
\end{definition}

The $\bot$ symbol here represents an ``I don't know'' response of the algorithm while an index $j$ output by $\mathcal{D}_S(S_\tau)$ for the $i^{th}$ symbol of $S_\tau$ should be interpreted as the $(n,\delta)$-indexing algorithm guessing that this was the $j^{th}$ symbol of $S$. One seeks algorithms that decode as many indices as possible correctly. Naturally, one can only \emph{correctly decode} indices that were \emph{correctly transmitted}. Next we give formal definitions of both notions:

\begin{definition}[Correctly Decoded Index]
An $(n,\delta)$ indexing algorithm $(S, \mathcal{D}_S)$ decodes index $j$ correctly under $\tau$ if $\mathcal{D}_S(S_\tau)$ outputs $i$ and there exists a $k$ such that $i = |del(\tau_1[1,k])|,~j = |del(\tau_2[1,k])|,~\tau_1[k] = S[i],~\tau_2[k] = S_\tau[j]$
\end{definition}

We remark that this definition counts any $\bot$ response as an incorrect decoding.

\fullOnly{\begin{definition}[Successfully Transmitted Symbol]}
\shortOnly{\begin{definition}[Successfully Transmitted]}
For string $S_\tau$, which was derived from a synchronization string $S$ via $\tau = (\tau_1,\tau_2)$, we call the $j^{th}$ symbol $S_\tau[j]$ successfully transmitted if it stems from a symbol coming from $S$, i.e., if there exists a $k$ such that $|del(\tau_2[1,k])|=j$ and $\tau_1[k] = \tau_2[k]$.
\end{definition}

We now define the quality of an $(n,\delta)$-indexing algorithm by counting the maximum number of misdecoded indices among those that were successfully transmitted. Note that the trivial indexing strategy with $S = 1,\ldots,n$ which outputs for each symbol the symbol itself has no misdecodings. One can therefore also interpret our quality definition as capturing how far from this ideal solution an algorithm is (stemming likely due to the smaller alphabet which is used for $S$).

\fullOnly{\begin{definition}[Misdecodings of an $(n,\delta)$-Indexing Algorithm]}
\shortOnly{\begin{definition}[Misdecodings]}
Let $(S, \mathcal{D}_S)$ be an $(n,\delta)$-indexing algorithm. We say this algorithm has at most $k$ misdecodings if for any $\tau$ corresponding to at most $n \delta$ insertions and deletions the number of correctly transmitted indices that are incorrectly decoded is at most $k$.
\end{definition}

\shortOnly{In Appendix~\ref{app:MoreOnIndexing}, we will define two important characteristics of the solutions of an indexing problem and provide a Table summarizing all solutions to the indexing problem which are going to be discussed in this paper.}

\global\def\FurtherOnIndexing
{

Now, we introduce two further useful properties that a $(n,\delta)$-indexing algorithm might have. 

\begin{definition}[Error-free Solution]
We call $(S, \mathcal{D}_S)$ an error-free $(n, \delta)$-indexing algorithm with respect to a set of deletion or insertion patterns if every index output is either $\bot$ or correctly decoded. In particular, the algorithm never outputs an incorrect index, even for indices which are not correctly transmitted. 
\end{definition}

It is noteworthy that error-free solutions are essentially only obtainable when dealing with the insertion-only or deletion-only setting. In both cases, the trivial solution with $S = 1, \cdots, n$ which decodes any index that was received exactly once is error-free. We later give some algorithms which preserve this nice property, even over a smaller alphabet, and show how error-freeness can be useful in the context of error correcting codes. 

Lastly, another very useful property of some $(n, \delta)$-indexing algorithms is that their decoding process operates in a streaming manner, i.e, the decoding algorithm decides the index output for $S_\tau[j]$ independently of $S_\tau[j']$ where $j' > j$. While this property is not particularly useful for the error correcting block code application put forward in this paper, it is an extremely important and strong property which is crucial in several applications we know of, such as, rateless error correcting codes, channel simulations, interactive coding, edit distance tree codes, and other settings. 

\begin{definition}[Streaming Solutions]
We call $(S, \mathcal{D}_S)$ a streaming solution if the decoded index for the $i$th element of the received string $S_\tau$ only depends on $S_\tau[1, i]$.
\end{definition}

Again, the trivial solution for $(n,\delta)$-index decoding problem over an alphabet of size $n$ with zero misdecodings can be made streaming by outputting for every received symbols the received symbol itself as an index. This solution is also error-free for the deletion-only setting but not error-free for the insertion-only setting. In fact, it is easy to show that an algorithm cannot be both streaming and error-free in any setting which allows insertions. 

\smallskip

Overall, the important characteristics of an $(n,\delta)$-indexing algorithm are (a) its alphabet size $|\Sigma|$, (b) the bound on the number of misdecodings, (c) the complexity of the decoding algorithm $\mathcal{D}$, (d) the preprocessing complexity of constructing the string $S$, (e) whether the algorithm works for the insertion-only, the deletion-only or the full insdel setting, and (f) whether the algorithm satisfies the streaming or error-freeness property. Table~\ref{tbl:indexingSolutions} gives a summary over the different solutions for the $(n,\delta)$-indexing problem we give in this paper. 

\begin{table*}[h]
\centering
 \begin{tabular}{|| c | c r c c c||} 
 \hline
 Algorithm & Type & Misdecodings & {Error-free} & Streaming & Complexity \\ [0.5ex] 
 \hline\hline
Section \ref{sec:sync_decoding} & ins/del & $(2 + \eps) \cdot n\delta$ &   & \checkmark & $O(n^4)$ \\
Section \ref{sec:indelErrors} & ins/del & $\sqrt{\eps} \cdot  n\,\ $   &  &  & $O\left(n^2/\sqrt{\eps}\right)$ \\
Section \ref{sec:delOnly} & del 		& $\eps \cdot n\delta$ &  & \checkmark & $O(n)$\\
Section \ref{sec:insErrors} & ins & $(1+\eps) \cdot n\delta$ & \checkmark &  & $O(n)$ \\
Section \ref{sec:insErrors} & del 		& $\eps \cdot  n\delta$ & \checkmark &  & $O(n)$ \\
Section \ref{sec:improvedStreaming} & ins/del & $(1+\eps) \cdot n\delta$ &  & \checkmark & $O(n^4)$ \\
 \hline
\end{tabular}
\caption{Properties and quality of $(n,\delta)$-indexing algorithms with $S$ being a $\eps$-synchronization string}\label{tbl:indexingSolutions}
\end{table*}
}

\tempfullOnly{\FurtherOnIndexing}

\section{Insdel Codes via Indexing Algorithms}\label{sec:codings}
Next, we show how a good $(n,\delta)$-indexing algorithms $(S, \mathcal{D}_S)$ over alphabet $\Sigma_S$ allows one to transform any regular ECC $\mathcal{C}$ with block length $n$ over alphabet $\Sigma_{\mathcal{C}}$ which can efficiently correct half-errors, i.e., symbol corruptions and erasures, into a good insdel code over alphabet $\Sigma = \Sigma_{\mathcal{C}} \times \Sigma_S$.

To this end, we simply attach $S$ symbol-by-symbol to every codeword of $\mathcal{C}$. On the decoding end, we first decode the indices of the symbols arrived using the indexing part of each received symbol and then interpret the message parts as if they have arrived in the decoded order. Indices where zero or multiple symbols are received get considered as erased. We will refer to this procedure as \emph{the indexing procedure}. Finally, the decoding algorithm $\mathcal{D}_{\mathcal{C}}$ for $\mathcal{C}$ is used. These two straight forward algorithms are formally described as Algorithm~\ref{alg:NewECCsEncoder} and Algorithm~\ref{alg:NewECCsDecoder}\shortOnly{ in Appendix~\ref{app:ECCviaIndexing}. The missing proofs in this section are available in Appendix~\ref{app:ECCviaIndexing}}.

\begin{theorem}\label{thm:MisDecodeToHalfError}
If $(S, \mathcal{D}_S)$ guarantees $k$ misdecodings for the $(n, \delta)$-index problem, then the indexing procedure recovers the codeword sent up to $n\delta + 2k$ half-errors, i.e., half-error distance of the sent codeword and the one recovered by the indexing procedure is at most $n\delta + 2k$. If $(S, \mathcal{D}_S)$ is error-free, the indexing procedure recovers the codeword sent up to $n\delta + k$ half-errors. 
\end{theorem}

\global\def\ProofOfMisDecodeToHalfError{
\fullOnly{\begin{proof}}
\shortOnly{\begin{proof}[of Theorem~\ref{thm:MisDecodeToHalfError}]}
Consider a set insertions and deletions described by $\tau$ consisting of $D_\tau$ deletions and $I_\tau$ insertions. Note that among $n$ encoded symbols, at most $D_\tau$ were deleted and less than $k$ of are decoded incorrectly. Therefore, at least $n - D_\tau - k$ indices are decoded correctly. On the other hand at most $D_\tau + k$ of the symbols sent are not decoded correctly. Therefore, if the output only consisted of correctly decoded indices for successfully transmitted symbols, the output would have contained up to $D_\tau + k$ erasures and no symbol corruption, resulting into a total of $D_\tau + k$ half-errors. However, any symbol which is being incorrectly decoded or inserted may cause a correctly decoded index to become an erasure by making it appear multiple times or change one of original $I_\tau + k$ erasures into a corruption error by making the indexing procedure mistakenly decode an index. Overall, this can increase the number of half-errors by at most $I_\tau + k$ for a total of at most 
$D_\tau + k + I_\tau + k = D_\tau + I_\tau + 2k = n\delta + 2k$ half-errors. For error-free indexing algorithms, any misdecoding does not result in an incorrect index and the number of incorrect indices  is $I_\tau$ instead of $I_\tau + k$ leading to the reduced number of half-errors in this case. 
\end{proof}
}

\fullOnly{\ProofOfMisDecodeToHalfError}

This makes it clear that applying an ECC $\mathcal{C}$ which is resilient to $n\delta + 2k$ half-errors enables the receiver side to fully recover $m$. 

\global\def\EncDecAlgorithms{
\begin{algorithm}
\caption{Insertion Deletion Encoder}
\begin{algorithmic}[1]\label{alg:NewECCsEncoder}
\REQUIRE $n$, $m=m_1,\cdots,m_n$

\STATE $\tilde m = \mathcal{E}_{\mathcal{C}} (m)$
\FOR {$\texttt{i} = 1$ to  $n$}
\STATE $M_i = (m_i, S[i])$ 
\ENDFOR
\medskip
\ENSURE $M$
\end{algorithmic}
\end{algorithm}

\begin{algorithm}
\caption{Insertion Deletion Decoder}
\begin{algorithmic}[1]\label{alg:NewECCsDecoder}
\REQUIRE $n$, $M' = (\tilde m', S')$

\STATE $Dec \leftarrow \mathcal{D}_{S}(S')$
\FOR {$\texttt{i} = 1$ to  $n$}
\IF {there is a unique $j$ for which $Dec[j] = i$}
\STATE $m'_i = \tilde m'_j$ 
\ELSE
\STATE $m'_i = $ ?
\ENDIF
\ENDFOR

\STATE $m = \mathcal{D}_{\mathcal{C}} (m')$
\ENSURE $m$
\end{algorithmic}
\end{algorithm}
}
\fullOnly{\EncDecAlgorithms}

Next, we formally state how a good $(n,\delta)$-indexing algorithm $(S, \mathcal{D}_S)$ over alphabet $\Sigma_S$ allows one to transform any regular ECC $\mathcal{C}$ with block length $n$ over alphabet $\Sigma_{\mathcal{C}}$ which can efficiently correct half-errors, i.e., symbol corruptions and erasures, into a good insdel code over alphabet $\Sigma = \Sigma_{\mathcal{C}} \times \Sigma_S$. The following Theorem is a corollary of Theorem~\ref{thm:MisDecodeToHalfError} and the definition of the indexing procedure:

\begin{theorem}\label{thm:mainECC}
Given an (efficient) $(n, \delta)$-indexing algorithm $(S, \mathcal{D}_S)$ over alphabet $\Sigma_S$ with at most $k$ misdecodings, and decoding complexity $T_{\mathcal{D}_{S}}(n)$ and an (efficient) ECC $\mathcal{C}$ over alphabet $\Sigma_{\mathcal{C}}$ with rate $R_{\mathcal{C}}$, encoding complexity $T_{\mathcal{E}_{\mathcal{C}}}$, and decoding complexity $T_{\mathcal{D}_{\mathcal{C}}}$ that corrects up to $n\delta + 2k$ half-errors, one obtains an insdel code that can be (efficiently) decoded from up to $n\delta$ insertions and deletions. The rate of this code is
$$R_{\mathcal{C}} \cdot \left(1 - \frac{\log \Sigma_S}{\log \Sigma_{\mathcal{C}}}\right)$$
The encoding complexity remains $T_{\mathcal{E}_{\mathcal{C}}}$, the decoding complexity is $T_{\mathcal{D}_{\mathcal{C}}} + T_{\mathcal{D}_{S}}(n)$ and the preprocessing complexity of constructing the code is the complexity of constructing $\mathcal{C}$ and $S$.\\
Furthermore, if $(S, \mathcal{D}_S)$  is error-free, then choosing a $\mathcal{C}$ which can recover only from $n\delta + k$ erasures is sufficient to produce the same quality code. 
\end{theorem}

Note that if one chooses $\Sigma_{\mathcal{C}}$ such that $\frac{\log \Sigma_S}{\log \Sigma_{\mathcal{C}}} = o(\delta)$, the rate loss due to the attached symbols will be negligible. With all this in place one can obtain Theorem~\ref{thm:main} as a consequence of Theorem~\ref{thm:mainECC}.

\fullOnly{\begin{proof}[Proof of Theorem~\ref{thm:main}]}
\shortOnly{\begin{proof}[of Theorem~\ref{thm:main}]}
Given the $\delta$ and $\eps$ from the statement of Theorem~\ref{thm:main} we choose $\eps' = O\left(\left(\frac{\eps}{6}\right)^2\right)$ and use Theorem~\ref{thm:detepsselfconstruction} to construct a string $S$ of length $n$ over alphabet $\Sigma_S$ of size $\eps^{-O(1)}$ with the $\eps'$-self-matching property. We then use the $(n, \delta)$-indexing algorithm $(S, \mathcal{D}_S)$ where given in Section~\ref{sec:indelErrors} and line 2 of Table~\ref{tbl:indexingSolutions} which guarantees that it has at most $\sqrt{\eps'} = \frac{\eps}{3}$ misdecodings. Finally, we choose a near-MDS expander code~\cite{guruswami2005linear} $\mathcal{C}$ which can efficiently correct up to $\delta_{\mathcal{C}} = \delta + \frac{\eps}{3}$ half-errors and has a rate of $R_{\mathcal{C}} > 1 - \delta_{\mathcal{C}} - \frac{\eps}{3}$ over an alphabet $\Sigma_{\mathcal{C}} = \exp(\eps^{-O(1)})$ such that $\log |\Sigma_{\mathcal{C}}|\geq \frac{3 \log |\Sigma_S|}{\eps}$. This ensures that the final rate is indeed at least $R_{\mathcal{C}} - \frac{\log \Sigma_S}{\log \Sigma_{\mathcal{C}}} = 1 - \delta - 3 \frac{\eps}{3}$ and the number of insdel errors that can be efficiently corrected is $\delta_{\mathcal{C}} - 2 \frac{\eps}{3} \geq \delta$. The encoding and decoding complexities are furthermore straight forward and as is the polynomial time preprocessing time given Theorem~\ref{thm:detepsselfconstruction} and \cite{guruswami2005linear}.
\end{proof}

\section{Synchronization Strings}\label{sec:sync}

In this section, we formally define and develop $\eps$-synchronization strings, which can be used as our base synchronization string $S$ in our $(n,\delta)$-indexing algorithms.

As explained in Section~\ref{sec:introsynchstrings} it makes sense to think of the prefixes $S[1,l]$ of a synchronization string $S$ as \emph{codewords} encoding their length $l$, as the prefix $S[1,l]$, or a corrupted version of it, will be exactly all the indexing information that has been received by the time the $l^{th}$ symbol is communicated:

\shortOnly{\begin{definition}[Associated Codewords]}
\fullOnly{\begin{definition}[Codewords Associated with a Synchronization String]}
Given any synchronization string $S$ we define the set of codewords associated with $S$ to be the set of prefixes of $S$, i.e., $\{ S[1,l] \  |  \ 1 \leq l \leq |S|\}$.
\end{definition}

Next, we define a distance metric on any set of strings, which will be useful in quantifying how good a synchronization string $S$ and its associated set of codewords is:

\begin{definition}[Relative Suffix Distance]
For any two strings $S, S' \in \Sigma^*$ we define their relative suffix distance $RSD$ as follows:
$$RSD(S,S') = \max_{k > 0} \frac{ED\left(S(|S|-k,|S|],S'(|S'|-k,|S'|]\right)}{2k}$$
\end{definition}

Next we show that RSD is indeed a distance which satisfies all properties of a metric for any set of strings.
To our knowledge, this metric is new. It is, however, similar in spirit to the suffix ``distance'' defined in~\cite{braverman2015coding}, which unfortunately is non-symmetric and does not satisfy the triangle inequality but can otherwise be used in a similar manner as RSD in the specific context here (see also Section~\ref{sec:improvedStreaming}).

\begin{lemma}\label{lem:RSDMetricProperties}
For any strings $S_1,S_2,S_3$ we have 	
\begin{itemize}
	\item {\bfseries Symmetry:} $RSD(S_1,S_2) = RSD(S_2,S_1)$,
	\item {\bfseries Non-Negativity and Normalization:} $0 \leq RSD(S_1,S_2) \leq 1$,
	\item {\bfseries Identity of Indiscernibles:} $RSD(S_1,S_2) = 0 \Leftrightarrow S_1 = S_2$, and
	\item {\bfseries Triangle Inequality:} $RSD(S_1,S_3) \leq RSD(S_1,S_2) + RSD(S_2,S_3)$.
\end{itemize}
In particular, RSD defines a metric on any set of strings. 
\end{lemma}
\shortOnly{The missing proofs in this section are available in Appendix~\ref{app:chapter3}.}

\global\def\ProofOfMetricPropertiesOfRSD{
\shortOnly{\begin{proof}[of Lemma~\ref{lem:RSDMetricProperties}]}
\fullOnly{\begin{proof}}
Symmetry and non-negativity follow directly from the symmetry and non-negativity of edit distance. Normalization follows from the fact that the edit distance between two length $k$ strings can be at most $2k$. To see the identity of indiscernibles note that $RSD(S_1,S_2) = 0$ if and only if for all $k$ the edit distance of the $k$ prefix of $S_1$ and $S_2$ is zero, i.e., if for every $k$ the $k$-prefix of $S_1$ and $S_2$ are identical. This is equivalent to $S_1$ and $S_2$ being equal. Lastly, the triangle inequality also essentially follows from the triangle inequality for edit distance. To see this let $\delta_1 = RSD(S_1,S_2)$ and $\delta_2 = RSD(S_2,S_3)$. By the definition of RSD this implies that for all $k$ the $k$-prefixes of $S_1$ and $S_2$ have edit distance at most $2 \delta_1 k$ and the $k$-prefixes of $S_2$ and $S_3$ have edit distance at most $2 \delta_2 k$. By the triangle inequality for edit distance, this implies that for every $k$ the $k$-prefix of $S_1$ and $S_3$ have edit distance at most $(\delta_1 + \delta_2) \cdot 2 k$ which implies that $RSD(S_1,S_3) \leq \delta_1 + \delta_2$.
\end{proof}}

\fullOnly{\ProofOfMetricPropertiesOfRSD}

With these definitions in place, it remains to find synchronization strings whose prefixes induce a set of codewords, i.e., prefixes, with large RSD distance. It is easy to see that the RSD distance for any two strings ending on a different symbol is one. This makes the trivial synchronization string, which uses each symbol in $\Sigma$ only once, induce an associated set of codewords of optimal minimum-RSD-distance one. Such trivial synchronization strings, however, are not interesting as they require an alphabet size linear in the length $n$. To find good synchronization strings over constant size alphabets, we give the following important definition of an $\eps$-synchronization string. The parameter $0<\eps<1$ should be thought of measuring how far a string is from the perfect synchronization string, i.e., a string of $n$ distinct symbols.

\begin{definition}[$\eps$-Synchronization String]\label{def:synCode}
String $S \in \Sigma^n$ is an $\eps$-synchronization string if for every $1 \leq i < j < k \leq n + 1$ we have that $ED\left(S[i, j),S[j, k)\right) > (1-\eps) (k-i)$. We call the set of prefixes of such a string an $\eps$-synchronization string. 
\end{definition}

The next lemma shows that the $\eps$-synchronization string property is strong enough to imply a good minimum RSD distance between any two codewords associated with it. 

\begin{lemma}\label{lem:codewordRSDdistance}
If $S$ is an $\eps$-synchronization string, then $RSD(S[1,i],S[1,j]) > 1-\eps$ for any $i < j$, i.e.,  any two codewords associated with $S$ have RSD distance of at least $1-\eps$. 
\end{lemma}
\begin{proof}
Let $k = j - i$. The $\eps$-synchronization string property of $S$ guarantees that 
\noSTOC{\[ED\left(S[i-k, i),S[i,j)\right) > (1-\eps) 2k.\]}
\STOConly{$ED\left(S[i-k, i),S[i,j)\right) > (1-\eps) 2k$.}
Note that this holds even if $i-k<1$. To finish the proof we note that the maximum in the definition of RSD includes the term $\frac{ED\left(S[i-k, i),S[i,j)\right)}{2k} > 1-\eps$, which implies that $RSD(S[1,i],S[1,j]) > 1-\eps$.
\end{proof}

\subsection{Existence and Construction}\label{sec:EpsSynchExistence}

The next important step is to show that the $\eps$-synchronization strings we just defined exist, particularly, over alphabets whose size is independent of the length $n$. We show the existence of $\eps$-synchronization strings of arbitrary length for any $\eps>0$ using an alphabet size which is only polynomially large in $1/\eps$. We remark that $\eps$-synchronization strings can be seen as a strong generalization of square-free sequences in which any two neighboring substrings $S[i, j)$ and $S[j, k)$ only have to be different and not also far from each other in edit distance. Thue \cite{thue1977uber} famously showed the existence of arbitrarily large square-free strings over a trinary alphabet. Thue's methods for constructing such strings however turns out to be fundamentally too weak to prove the existence of $\eps$-synchronization strings, for any constant $\eps<1$.

\global\def\StatementofLemLLL{
\begin{lemma}[General Lov\'{a}sz local lemma]\label{lem:LLL}
Let $A_1, \dots, A_n$ be a set of ``bad'' events. The directed graph $G(V, E)$ is called a dependency graph for this set of events if $V = \{1, \dots, n\}$ and each event $A_i$ is mutually independent of all the events $\{A_j: (i, j) \not\in E\}$. 

Now, if there exists $x_1, \dots, x_n \in [0,1)$ such that for all $i$ we have 
$$\mathbb{P}\left[A_i\right] \leq x_i \prod_{(i, j) \in E}\left(1-x_j\right)$$
 then there exists a way to avoid all events $A_i$ simultaneously and the probability for this to happen is bounded by 
 $$\mathbb{P}\left[ \bigwedge_{i = 1}^n \bar{A}_i\right] \geq \prod_{i = 1}^n \left(1 - x_i\right) > 0.$$
\end{lemma}
}

\fullOnly{\noSTOC{\smallskip Our existence proof requires the general Lov\'{a}sz local lemma which we recall here first:\StatementofLemLLL}}

\begin{theorem}\label{thm:existence}
For any $\eps\in (0,1)$, $n \geq 1$, there exists an $\eps$-synchronization string of length $n$ over an alphabet of size $\Theta(1/\eps^4)$.
\end{theorem}

\global\def\ProofSketchofThmExistence{
\begin{proof}[Sketch]
Let $S$ be a string of length $n$ obtained by concatenating two strings $T$ and $R$, where $T$ is simply the repetition of $0, \dots, t-1$ for $t = \Theta(1/\eps^2)$, and $R$ is a uniformly random string of length $n$ over alphabet $\Sigma$. In particular, $S_i = \left(i \bmod t, R_i\right)$. We use the probabilistic method to prove that with positive probability $S$ does not contain \emph{bad triple}, where $(x, y, z)$ is a bad triple if $ED(S[x, y), S[y, z)) \le (1-\eps) (x-z)$. In particular, triples with $x-z<t$ are never bad because of $T$ and a counting argument enumerating all possible insertion-deletion patterns together with a Chernoff bound assessing the likelihood of a fixed such pattern to be valid shows that the probability for a fixed triple with $x-z>t$ to be bad is exponentially small in $x-z$. It is furthermore easy to see that the probability of two triples $(x,y,z)$ and $(x',y',z')$ being bad is independent unless they overlap. This shows a sparse dependency structure between the bad events associated which each triple being bad which is strong enough to apply the general Lov\'{a}sz local lemma to show the desired positive probability of avoiding all violations of the $\eps$-synchronization property. The complete proof and a statement of the general Lov\'{a}sz local lemma can be found in Appendix~\ref{app:chapter3}. 
\end{proof}
}

\global\def\ProofofThmExistence{
\fullOnly{\begin{proof}}
\shortOnly{\begin{proof}[of Theorem~\ref{thm:existence}]}
Let $S$ be a string of length $n$ obtained by concatenating two strings $T$ and $R$, where $T$ is simply the repetition of $0, \dots, t-1$ for $t = \Theta\left(\frac{1}{\eps^2}\right)$, and $R$ is a uniformly random string of length $n$ over alphabet $\Sigma$. In particular, $S_i = \left(i \bmod t, R_i\right)$.

We prove that $S$ is an $\eps$-synchronization string by showing that there is a positive probability that $S$ contains no \emph{bad triple}, where $(x, y, z)$ is a bad triple if $ED(S[x, y), S[y, z)) \le (1-\eps) (z-x)$.

First, note that a triple $(x, y, z)$ for which $z - x < t$ cannot be a bad triple as it consists of completely distinct symbols by courtesy of $T$. Therefore, it suffices to show that there is no bad triple $(x, y, z)$ in $R$ for $x, y, z$ such that $z - x > t$. 

Let $(x, y, z)$ be a bad triple and let $a_1 a_2 \cdots a_k$ be the longest common subsequence of $R[x, y)$ and $R[y, z)$. It is straightforward to see that $ED(R[x, y), R[y, z)) = (y - x) + (z - y) - 2k = z - x - 2k.$ Since $(x, y, z)$ is a bad triple, we have that $z - x - 2k \le (1-\eps) (z-x)$, which means that $k \ge \frac{\eps}{2} (z-x)$. With this observation in mind, we say that $R[x, z)$ is a \emph{bad interval} if it contains a subsequence $a_1 a_2 \cdots a_k a_1 a_2 \cdots a_k$ such that $k \ge \frac{\eps}{2} (z-x)$.

To prove the theorem, it suffices to show that a randomly generated string does not contain any bad intervals with a non-zero probability. 
We first upper bound the probability that an interval of length $l$ is bad:
\noSTOC{
\begin{eqnarray*}
\Pr_{I \sim \Sigma^l}[I\text{ is bad}] &\le& {l \choose {\eps l}} {|\Sigma|}^{-\frac{\eps l}{2}}\\
&\le& \left(\frac{el}{\eps l}\right)^{\eps l} {|\Sigma|}^{-\frac{\eps l}{2}}\\
&=& \left(\frac{e}{\eps \sqrt{|\Sigma|}}\right)^{\eps l},
\end{eqnarray*}}
\STOConly{
\begin{eqnarray*}
\Pr_{I \sim \Sigma^l}[I\text{ is bad}] &\le& {l \choose {\eps l}} {|\Sigma|}^{-\frac{\eps l}{2}}
\le \left(\frac{el}{\eps l}\right)^{\eps l} {|\Sigma|}^{-\frac{\eps l}{2}}
= \left(\frac{e}{\eps \sqrt{|\Sigma|}}\right)^{\eps l},
\end{eqnarray*}}
where the first inequality holds because if an interval of length $l$ is bad, then it must contain a repeating subsequence of length $\frac{l\eps}{2}$. Any such sequence can be specified via $\eps l$ positions in the $l$ long interval and the probability that a given fixed sequence is valid for a random string is ${|\Sigma|}^{-\frac{\eps l}{2}}$. The second inequality comes from the fact that ${n \choose k} < \left(\frac{ne}{k}\right)^k$.

The resulting inequality shows that the probability of an interval of length $l$ being 	bad is bounded above by $C^{-\eps l}$, where $C$ can be made arbitrarily large by taking a sufficiently large alphabet size $|\Sigma|$.

To show that there is a non-zero probability that the uniformly random string $R$ contains no bad interval $I$ of size $t$ or larger, we use the general Lov\'{a}sz local lemma\STOConly{.}\noSTOC{ stated in Lemma~\ref{lem:LLL}.}  Note that the badness of interval $I$ is mutually independent of the badness of all intervals that do not intersect $I$. We need to find real numbers $x_{p, q}\in[0,1)$ corresponding to intervals $R[p, q)$ for which $$\\Pr\left[\text{Interval }R[p, q)\text{ is bad}\right] \le x_{p, q} \prod_{R[p, q)\cap R[p', q')\not = \emptyset} (1-x_{p', q'}).$$

We have seen that the left-hand side can be upper bounded by $C^{-\eps\left|R[p, q)\right|} = C^{\eps(p-q)}$. 
Furthermore, any interval of length $l'$ intersects at most $l + l'$ intervals of length $l$. 
We propose $x_{p, q} = D^{-\eps\left|R[p, q)\right|} = D^{\eps(p-q)}$ for some constant $D > 1$.
This means that it suffices to find a constant $D$ that for all substrings $R[p, q)$ satisfies
$$C^{\eps(p-q)} \le D^{\eps(p-q)} \prod_{l=t}^n \left(1-D^{-\eps l}\right)^{l + (q-p)},$$
or more clearly, for all $l' \in \{1,\cdots,n\}$,
$$C^{-l'} \le D^{-l'} \prod_{l=t}^n \left(1-D^{-\eps l}\right)^\frac{l + l'}{\eps},$$
which means that
\begin{equation}\label{eqn:LLL1}
C \ge \frac{D}{\prod_{l=t}^n \left(1-D^{-\eps l}\right)^{\frac{1+l/l'}{\eps}} }.
\end{equation}

For $D>1$, the right-hand side of Equation~\eqref{eqn:LLL1} is maximized when $n = \infty$ and $l' = 1$, and since we want Equation \eqref{eqn:LLL1} to hold for all $n$ and all $l' \in\{1,\cdots, n\}$, it suffices to find a $D$ such that
\begin{equation}\nonumber
C \ge \frac{D}{\prod_{l=t}^\infty \left(1-D^{-\eps l}\right)^\frac{l+1}{\eps} }.
\end{equation}

To this end, let
$$L = \min_{D>1} \left\{   \frac{D}{\prod_{l=t}^\infty \left(1-D^{-\eps l}\right)^\frac{l+1}{\eps} }    \right\}.$$

Then, it suffices to have $\Sigma$ large enough so that
\[C = \frac{\eps\sqrt{|\Sigma|}}{e} \ge L,\] which means that $|\Sigma| \ge \frac{e^2L^2}{\eps^2} $ suffices to allow us to use the Lov\'{a}sz local lemma. We claim that $L = \Theta(1)$, 
which will complete the proof.
Since $t = \omega\left(\frac{\log\frac{1}{\eps}}{\eps}\right)$, 
$$\forall l \ge t \qquad D^{-\eps l} \cdot \frac{l+1}{\eps} \ll 1.$$

Therefore, we can use the fact that $(1-x)^k > 1 - xk$ to show that:
\begin{eqnarray}
\frac{D}{\prod_{l=t}^\infty \left(1-D^{-\eps l}\right)^\frac{l+1}{\eps} } &<& \frac{D}{\prod_{l=t}^\infty \left(1-\frac{l+1}{\eps}\cdot D^{-\eps l}\right) }\\
&<&\frac{D}{1-\sum_{l=t}^\infty  \frac{l+1}{\eps}\cdot D^{-\eps l} }\label{eqn:linearize}\\
&=&\frac{D}{1-\frac{1}{\eps}\sum_{l=t}^\infty  (l+1)\cdot \left(D^{-\eps}\right)^l }\\
&=&\frac{D}{1-\frac{1}{\eps} \frac{2t\left(D^{-\eps}\right)^{t}}{(1-D^{-\eps})^2}}\label{eqn:PowerSeries}\\
&=&\frac{D}{1- \frac{2}{\eps^3} \frac{D^{-\frac{1}{\eps}}}{(1-D^{-\eps})^2}}.\label{eqn:upperbound}
\end{eqnarray}

Equation \eqref{eqn:linearize} is derived using the fact that $\prod_{i=1}^\infty(1-x_i) \ge 1-\sum_{i=1}^\infty x_i$ and Equation \eqref{eqn:PowerSeries} is a result of the following equality for $x<1$:
$$\sum_{l=t}^{\infty} (l+1)x^l = \frac{x^t (1+t-t x)}{(1-x)^2} < \frac{2t x^t}{(1-x)^2}.$$	

One can see that for $D=7	$, $\max_\eps \left\{\frac{2}{\eps^3} \frac{D^{-\frac{1}{\eps}}}{(1-D^{-\eps})^2}\right\} < 0.9$, and therefore step \eqref{eqn:linearize} is legal and \eqref{eqn:upperbound} can be upper-bounded by a constant. Hence, $L=\Theta(1)$ and the proof is complete. \end{proof}
}

\fullOnly{\ProofofThmExistence}
\shortOnly{\ProofSketchofThmExistence}

\global\def\AlphabetSizeRemarks
{
\medskip
\noindent{\bfseries Remarks on the alphabet size:} 
Theorem~\ref{thm:existence} shows that for any $\eps >0$ there exists an $\eps$-synchronization string over alphabets of size $O(\eps^{-4})$. A polynomial dependence on $\eps$ is also necessary. In particular, there do not exist any $\eps$-synchronization string over alphabets of size smaller than $\eps^{-1}$. In fact, any consecutive substring of size $\eps^{-1}$ of an $\eps$-synchronization string has to contain completely distinct elements. This can be easily proven as follows: For sake of contradiction let $S[i, i+\eps^{-1})$ be a substring of an $\eps$-synchronization string where $S[j] = S[j']$ for $i\leq j<j'< i+\eps^{-1}$. Then, $ED\left(S[j], S[j+1, j'+1))\right) = j'- j - 1 = (j'+ 1 - j) - 2 \leq (j'+ 1 - j) (1 - 2\eps)$. We believe that using the Lov\'{a}sz Local Lemma together with a more sophisticated non-uniform probability space, which avoids any repeated symbols within a small distance, allows avoiding the use of the string $T$ in our proof and improving the alphabet size to $O(\eps^{-2})$. It seems much harder to improved the alphabet size to $o(\eps^{-2})$ and we are not convinced that it is possible. This work thus leaves open the interesting question of closing the quadratic gap between $O(\eps^{-2})$ and $\Omega(\eps^{-1})$ from either side.
\medskip
}

\fullOnly{\AlphabetSizeRemarks}

Theorem~\ref{thm:existence} also implies an efficient randomized construction. 

\begin{lemma}\label{lem:randomizedpolytimeconstruction}
There exists a randomized algorithm which for any $\eps > 0$ constructs a $\eps$-synchronization string of length $n$ over an alphabet of size $O(\eps^{-4})$ in expected time $O(n^5)$. 
\end{lemma}	
\begin{proof}
Using the algorithmic framework for the Lov\'{a}sz local lemma given by Moser and Tardos~\cite{MoserTardos2010} and the extensions by Haeupler et al.\cite{HaeuplerJACM12p28} one can get such a randomized algorithm from the proof in Theorem~\ref{thm:existence}. The algorithm starts with a random string over any alphabet $\Sigma$ of size $\eps^{-C}$ for some sufficiently large $C$. It then checks all $O(n^2)$ intervals for a violation of the $\eps$-synchronization string property. For every interval this is an edit distance computation which can be done in $O(n^2)$ time using the classical Wagner-Fischer dynamic programming algorithm. If a violating interval is found the symbols in this interval are assigned fresh random values. This is repeated until no more violations are found. \cite{HaeuplerJACM12p28} shows that this algorithm performs only $O(n)$ expected number of re-samplings. This gives an expected running time of $O(n^5)$ overall, as claimed.
\end{proof}

\smallskip

Lastly, since synchronization strings can be encoded and decoded in a streaming fashion they have many important applications in which the length of the required synchronization string is not known in advance. In such a setting it is advantageous to have an infinite synchronization string over a fixed alphabet. In particular, since every consecutive substring of an $\eps$-synchronization string is also an $\eps$-synchronization string by definition, having an infinite $\eps$-synchronization string also implies the existence for every length $n$, i.e., Theorem~\ref{thm:existence}. Interestingly, a simple argument shows that the converse is true as well, i.e., the existence of an $\eps$-synchronization string for every length $n$ implies the existence of an infinite $\eps$-synchronization string over the same alphabet:

\begin{lemma}\label{lem:infinitesynchcodeexistence}
For any $\eps\in (0,1)$ there exists an infinite $\eps$-synchronization string over an alphabet of size $\Theta(1/\eps^4)$.
\end{lemma}

\global\def\ProofOfInfiniteSynchExistence
{
\fullOnly{\begin{proof}[Proof of Lemma~\ref{lem:infinitesynchcodeexistence}]}
\shortOnly{\begin{proof}[of Lemma~\ref{lem:infinitesynchcodeexistence}]}
Fix any $\eps \in (0, 1)$. According to Theorem~\ref{thm:existence} there exist an alphabet $\Sigma$ of size $O(1/\eps^4)$ such that there exists an at least one $\eps$-synchronization strings over $\Sigma$ for every length $n \in \mathbb{N}$. We will define a synchronization string $S = s_1 \cdot s_2 \cdot s_3 \ldots$ with $s_i \in \Sigma$ for any $i \in \mathbb{N}$ for which the $\eps$-synchronization property holds for any $i,j,k \in \mathbb{N}$. We define this string inductively. In particular, we fix an ordering on $\Sigma$ and define $s_1 \in \Sigma$ to be the first symbol in this ordering such that an infinite number of $\eps$-synchronization strings over $\Sigma$ starts with $s_1$. Given that there is an infinite number of $\eps$-synchronization over $\Sigma$ such an $s_1$ exists. Furthermore, the set of $\eps$-synchronization strings over $\Sigma$ which start with $s_1$ remains infinite by definition, allowing us to define $s_2 \in \Sigma$ to be the lexicographically first symbol in $\Sigma$ such there exists an infinite number of $\eps$-synchronization strings over $\Sigma$ starting with $s_1 \cdot s_2$. In the same manner, we inductively define $s_i$ to be the lexicographically first symbol in $\Sigma$ for which there exists and infinite number of $\eps$-synchronization strings over $\Sigma$ starting with $s_1 \cdot s_2 \cdot \ldots \cdot s_i$. To see that the infinite string defined in this manner does indeed satisfy the edit distance requirement of the $\eps$-synchronization property defined in Definition~\ref{def:synCode}, we note that for every $i<j<k$ with $i,j,k \in \mathbb{N}$ there exists, by definition, an $\eps$-synchronization string, and in fact an infinite number of them, which contains $S[1,k]$ and thus also $S[i,k]$ as a consecutive substring implying that indeed $ED\left(S[i, j),S[j, k)\right) > (1-\eps) (k-i)$ as required. Our definition thus produces the unique lexicographically first infinite $\eps$-synchronization string over $\Sigma$. 
\end{proof}

}

\fullOnly{\ProofOfInfiniteSynchExistence}

\smallskip

We remark that any string produced by the randomized construction of Lemma~\ref{lem:randomizedpolytimeconstruction} is guaranteed to be a correct $\eps$-synchronization string (not just with probability one). This randomized synchronization string construction is furthermore only needed once as a pre-processing step. The encoder or decoder of any resulting error correcting codes do not require any randomization. Furthermore, in Section~\ref{sec:improved_decoding} we will provide a deterministic polynomial time construction of a relaxed version of $\eps$-synchronization strings that can still be used as a basis for good $(n,\delta)$-indexing algorithms thus leading to insdel codes with a deterministic polynomial time code construction as well. 

\smallskip

It nonetheless remains interesting to obtain fast deterministic constructions of finite and infinite $\eps$-synchronization strings. In a subsequent work we achieve such efficient deterministic constructions for $\eps$-synchronization strings.  Our constructions even produce the infinite $\eps$-synchronization string $S$ proven to exist by Lemma~\ref{lem:infinitesynchcodeexistence}, which is much less explicit: While for any $n$ and any $\eps$ an $\eps$-synchronization string of length $n$ can in principle be found using an exponential time enumeration there is no straight forward algorithm which follows the proof of Lemma~\ref{lem:infinitesynchcodeexistence} and given an $i \in \mathbb{N}$ produces the $i^{th}$ symbol of such an $S$ in a finite amount of time (bounded by some function in $i$).
Our constructions require significantly more work but in the end lead to an explicit deterministic construction of an infinite $\eps$-synchronization string for any $\eps>0$ for which the $i^{th}$ symbol can be computed in only $O(\log i)$ time -- thus satisfying one of the strongest notions of constructiveness that can be achieved.

\subsection{Decoding}\label{sec:sync_decoding}

We now provide an algorithm for decoding synchronization strings, i.e., an algorithm that can form a solution to the indexing problem along with $\eps$-synchronization strings. In the beginning of Section~\ref{sec:sync}, we introduced the notion of relative suffix distance between two strings. 
Theorem~\ref{lem:codewordRSDdistance} stated a lower bound of $1-\eps$ for relative suffix distance between any two distinct codewords associated with an $\eps$-synchronization string, i.e., its prefixes. Hence, a natural decoding scheme for detecting the index of a received symbol would be finding the prefix with the closest relative suffix distance to the string received thus far. We call this algorithm \emph{the minimum relative suffix distance decoding algorithm}.

We define the notion of \emph{relative suffix error density at index $i$} which presents the maximized density of errors taken place over suffixes of $S[1, i]$.  We will introduce a very natural decoding approach for synchronization strings that simply works by decoding a received string by finding the codeword of a synchronization string $S$ (prefix of synchronization string) with minimum distance to the received string. We will show that this decoding procedure works correctly as long as the relative suffix error density is not larger than $\frac{1-\eps}{2}$. Then, we will show that if adversary is allowed to perform $c$ many insertions or deletions, the relative suffix distance may exceed $\frac{1-\eps}{2}$ upon arrival of at most $\frac{2c}{1-\eps}$ many successfully transmitted symbols. Finally, we will deduce that this decoding scheme decodes indices of received symbols correctly for all but $\frac{2c}{1-\eps}$ many of successfully transmitted symbols. Formally, we claim that:
\begin{theorem}\label{thm:MinRSDMisDec}
Any $\eps$-synchronization string of length $n$ along with the minimum relative suffix distance decoding algorithm form a solution to $(n, \delta)$-indexing problem that guarantees $\frac{2}{1-\eps}n\delta$ or less misdecodings. This decoding algorithm is streaming and can be implemented so that it works in $O(n^4)$ time.
\end{theorem}

\global\def\RSDMinDistanceDecoding{
\fullOnly{Before proceeding to the formal statement and the proofs of the claims above, we first provide the following useful definitions.}
\shortOnly{Before proceeding to the formal statement and proof of claims mentioned in Section~\ref{sec:EpsSynchExistence}, we first provide the following useful definitions.}

\begin{definition}[Error Count Function]
Let $S$ be a string sent over an insertion-deletion channel. We denote the \emph{error count from index $i$ to index $j$} with $\mathcal{E}(i, j)$ and define it to be the number of insdels applied to $S$ from the moment $S[i]$ is sent until the moment $S[j]$ is sent. $\mathcal{E}(i, j)$ counts the potential deletion of $S[j]$. However, it does \emph{not} count the potential deletion of $S[i]$.
\end{definition}

\begin{definition}[Relative Suffix Error Density]
Let string $S$ be sent over an insertion-deletion channel and let $\mathcal{E}$ denote the corresponding error count function. We define the \emph{relative suffix error density} of the communication as:
$$\max_{i\ge1} \frac{\mathcal{E}\left(|S|-i, |S|\right)}{i}$$
\end{definition}

The following lemma relates the suffix distance of the message being sent by sender and the message being received by the receiver at any point of a communication over an insertion-deletion channel to the relative suffix error density of the communication at that point.

\begin{lemma}\label{lem:SDDensity}
Let string $S$ be sent over an insertion-deletion channel and the corrupted message $S'$ be received on the other end. The relative suffix distance $RSD(S, S')$ between the string $S$ that was sent and the string $S'$ which was received is at most the relative suffix error density of the communication.
\end{lemma}
\begin{proof}
Let $\tilde\tau=(\tilde\tau_1, \tilde\tau_2)$ be the string matching from $S$ to $S'$ that 
characterizes insdels that have turned $S$ into $S'$.
Then:
\begin{eqnarray}
\shortOnly{&&}
RSD(S, S')
\shortOnly{\nonumber\\}
&=&  \max_{k > 0} \frac{ED(S(|S|-k,|S|],S'(|S'|-k,|S'|])}{2k}\label{eq:RES-RSE:step1}\\
&=& \max_{k > 0} \frac{\min_{\tau:S(|S|-k,|S|]\rightarrow S'(|S'|-k,|S'|]}\{sc(\tau_1) + sc(\tau_2)\}}{2k}\label{eq:RES-RSE:step2}\\
&\le& \max_{k > 0} \frac{2(sc(\tau'_1) + sc(\tau'_2))}{2k} 
\shortOnly{\\&}
\le \shortOnly{&}\text{Relative Suffix Error Density}\label{eq:RES-RSE:step3}
\end{eqnarray}
where $\tau'$ is $\tilde\tau$ limited to its suffix corresponding to $S(|S|-k, |S|])$. Note that Steps~\eqref{eq:RES-RSE:step1} and \eqref{eq:RES-RSE:step2} follow from the definitions of edit distance and relative suffix distance. Moreover, to see Step~\eqref{eq:RES-RSE:step3}, one has to note that one single insertion or deletion on the $k$-element suffix of a string may result into a string with $k$-element suffix of edit distance two of the original string's $k$-element suffix; one stemming from the inserted/deleted symbol and the other one stemming from a symbol appearing/disappearing at the beginning of the suffix in order to keep the size of suffix $k$.
\end{proof}

A key consequence of Lemma~\ref{lem:SDDensity} is that if an $\eps$-synchronization string is being sent over an insertion-deletion channel and at some step the relative suffix error density corresponding to corruptions is smaller than $\frac{1-\eps}{2}$, the relative suffix distance of the sent string and the received one at that point is smaller than $\frac{1-\eps}{2}$; therefore, as RSD of all pairs of codewords associated with an $\eps$-synchronization string are greater than $1-\eps$, the receiver can correctly decode the index of the corrupted codeword he received by simply finding the codeword with minimum relative suffix distance.

The following lemma states that such a guarantee holds most of the time during transmission of a synchronization string:
\begin{lemma}\label{lem:NumberofBadSuffixDistances}
Let $\eps$-synchronization string $S$ be sent over an insertion-channel channel and corrupted string $S'$ be received on the other end. If there are $c_i$ symbols inserted and $c_d$ symbols deleted, then, for any integer $t$, the relative suffix error density is smaller than $\frac{1-\eps}{t}$ upon arrival of all but $\frac{t(c_i+c_d)}{1-\eps} - c_d$ many of the successfully transmitted symbols.
\end{lemma}
\begin{proof}
Let $\mathcal{E}$ denote the error count function of the communication. We define the potential function $\Phi$ over $\{0, 1,\cdots,n\}$ as follows:
$$\Phi(i) = \max_{1\le s\le i} \left\{\frac{t\cdot \mathcal{E}(i-s, i)}{1-\eps} - s\right\}$$
Also, set $\Phi(0) = 0$. We prove the theorem by showing the correctness of the following claims:
\begin{enumerate}
\item If $\mathcal{E}(i-1, i) = 0$, i.e., the adversary does not insert or delete any symbols in the interval starting right after the moment $S[i-1]$ is sent and ending at when $S[i]$ is sent, then the value of $\Phi$ drops by 1 or becomes/stays zero, i.e., $\Phi(i) = \max\left\{0, \Phi(i-1)-1\right\}$.
\item If $\mathcal{E}(i-1, i) = k$, i.e., adversary inserts or deletes $k$ symbols in the interval starting right after the moment $S[i-1]$ is sent and ending at when $S[i]$ is sent, then the value of $\Phi$ increases by at most $\frac{tk}{1-\eps} - 1$, i.e., $\Phi(i) \le \Phi(i-1) +\frac{tk}{1-\eps}-1$.
\item If $\Phi(i) = 0$, then 
the relative suffix error density of the string that is received when $S[i]$ arrives at the receiving side is not larger than $\frac{1-\eps}{t}$.
\end{enumerate}

Given the correctness of claims made above, the lemma can be proved as follows. As adversary can apply at most $c_i+c_d$ insertions or deletions, $\Phi$ can gain a total increase of $\frac{t\cdot(c_i+c_d)}{1-\eps}$. Therefore, the value of $\Phi$ can be non-zero for at most $\frac{t\cdot (c_i+c_d)}{1-\eps}$ many inputs. As value of $\Phi(i)$ is non-zero for all $i$'s where $S[i]$ has been removed by adversary, there are at most $\frac{t\cdot (c_i+c_d)}{1-\eps} - c_d$ indices $i$ where $\Phi(i)$ is non-zero and $i$ is successfully transmitted. Hence, at most $\frac{t\cdot (c_i+c_d)}{1-\eps}-c_d$ many of correctly transmitted symbols can possibly be decoded incorrectly. 

We now proceed to the proof of each of the above claims to finish the proof:
\begin{enumerate}
\item In this case, $\mathcal{E}(i-s, i) = \mathcal{E}(i-s, i-1)$. So,
\begin{eqnarray*}
\Phi(i) &=& \max_{1\le s\le i} \left\{\frac{t\cdot\mathcal{E}(i-s, i)}{1-\eps} - s\right\} \\
&=& \max_{1\le s\le i} \left\{\frac{t\cdot\mathcal{E}(i-s, i-1)}{1-\eps} - s\right\} \\
&=& \max\left\{0, 
\max_{2\le s\le i} \left\{\frac{t\cdot\mathcal{E}(i-s, i-1)}{1-\eps} - s\right\}  
\right\} \\
&=& \max\left\{0, 
\max_{1\le s\le i-1} \left\{\frac{t\cdot\mathcal{E}(i-1-s, i-1)}{1-\eps} - s-1\right\}  
\right\} \\
&=& \max\left\{0, \Phi(i-1)-1\right\}
\end{eqnarray*}
\item In this case, $\mathcal{E}(i-s, i) = \mathcal{E}(i-s, i-1) + k$. So,
\begin{eqnarray*}
\Phi(i) &=& \max_{1\le s\le i} \left\{\frac{t\cdot\mathcal{E}(i-s, i)}{1-\eps} - s\right\} \\
&=& \max\left\{\frac{tk}{1-\eps}-1, 
\max_{2\le s\le i} \left\{\frac{t\cdot\mathcal{E}(i-s, i-1)+tk}{1-\eps} - s\right\}  
\right\} \\
&=& \max\left\{\frac{tk}{1-\eps}-1, 
\frac{tk}{1-\eps} + \max_{1\le s\le i-1} \left\{\frac{t\cdot\mathcal{E}(i-1-s, i-1)}{1-\eps} - s-1\right\}  
\right\} \\
&=& \frac{tk}{1-\eps} -1 + \max\left\{0, 
\max_{1\le s\le i-1} \left\{\frac{t\cdot\mathcal{E}(i-1-s, i-1)}{1-\eps} - s\right\}  
\right\}\\
&=& \frac{tk}{1-\eps} -1 + \max\left\{0, 
\Phi(i-1)\right\}\\
&=&\Phi(i-1) + \frac{tk}{1-\eps} -1
\end{eqnarray*}
\item 
\begin{eqnarray*}
&&\Phi(i) = \max_{1\le s\le i} \left\{\frac{t\cdot\mathcal{E}(i-s, i)}{1-\eps} - s\right\} = 0\\
&\Rightarrow &  \forall 1\le s \le i:~\frac{t\cdot\mathcal{E}(i-s, i)}{1-\eps} - s \le 0\\
&\Rightarrow &  \forall 1\le s \le i:~t\cdot\mathcal{E}(i-s, i)\le s(1-\eps)\\
&\Rightarrow &  \forall 1\le s \le i:~\frac{\mathcal{E}(i-s, i)}{s}\le \frac{\eps}{t}\\
&\Rightarrow &  \text{Relative Suffix Error Density } = \max_{1\le s \le i}\left\{\frac{\mathcal{E}(i-s, i)}{s}\right\}\le \frac{1-\eps}{t}\\
\end{eqnarray*}
\end{enumerate}
These finish the proof of the lemma.
\end{proof}

Now, we have all necessary tools to analyze the performance of the minimum relative suffix distance decoding algorithm: 

\fullOnly{\begin{proof}[Proof of Theorem~\ref{thm:MinRSDMisDec}]}
\shortOnly{\begin{proof}[of Theorem~\ref{thm:MinRSDMisDec}]}
As adversary is allowed to insert or delete up $n\delta$ symbols, by Lemma~\ref{lem:NumberofBadSuffixDistances}, there are at most $\frac{2n\delta}{1-\eps}$ successfully transmitted symbols during the arrival of which at the receiving side, the relative suffix error density is greater than $\frac{1-\eps}{2}$; Hence, by Lemma~\ref{lem:SDDensity}, there are at most $\frac{2n\delta}{1-\eps}$ misdecoded successfully transmitted symbols.

Further, we remark that this algorithm can be implemented in $O(n^4)$ as follows:	 
Using dynamic programming, we can pre-process the edit distance of any consecutive substring of $S$, like $S[i, j]$ to any  consecutive substring of $S'$, like $S'[i',j']$, in $O(n^4)$. Then, for each symbol of the received string, like $S'[l']$, we can find the codeword with minimum relative suffix distance to $S'[1,l']$ by calculating the relative suffix distance of it to all $n$ codewords. Finding suffix distance of $S'[1,l']$ and a codeword like $S[1,l]$ can also be simply done by minimizing $\frac{ED(S(l-k, l], S'(l'-k, l'])}{k}$ over $k$ which can be done in $O(n)$. With a $O(n^4)$ pre-process and a $O(n^3)$ computation as mentioned above, we have shown that the decoding process can be implemented in $O(n^4)$.
\end{proof}
}

\fullOnly{\RSDMinDistanceDecoding}
\shortOnly{
The details of the proof of Theorem~\ref{thm:MinRSDMisDec} are available in Appendix~\ref{app:MinRSDDecoding}.
\smallskip
}

We remark that by taking $\eps = o(1)$, one can obtain a solution to the $(n,\delta)$-indexing problem with a misdecoding guarantee of $2n\delta(1+o(1))$ which, using Theorem~\ref{thm:MisDecodeToHalfError}, results into a translation of $n\delta$ insertions and deletions into $n\delta(5+o(1))$ half-errors. As explained in Section~\ref{sec:introglobaldecoding}, such a guarantee however falls short of giving Theorem~\ref{thm:main}. 
In Section~\ref{sec:improvedStreaming}, we show that this guarantee of the min-distance-decoder can be slightly improved to work beyond an RSD distance of $\frac{1-\eps}{2}$, at the cost of some simplicity, by considering an alternative distance measure. In particular, the relative suffix pseudo distance RSPD, which was introduced in \cite{braverman2015coding}, can act as a metric stand-in for the minimum-distance decoder and lead to slightly improved decoding guarantees, despite neither being symmetric nor satisfying the triangle inequality. For any set of $k=k_i+k_d$ insdel errors consisting of $k_i$ insertions and $k_d$ deletions the RSPD based indexing solution leads to at most $(1 + \eps)(3k_i+k_d)$ half-errors which does imply ``near-MDS'' codes for deletion-only channels but still falls short for general insdel errors.

This leaves open the intriguing question whether a further improved (pseudo) distance definition can achieve an indexing solution with negligible number of misdecodings for the minimum-distance decoder.

\section{More Advanced Global Decoding Algorithms}\label{sec:improved_decoding}

Thus far, we have introduced $\eps$-synchronization strings as fitting solutions to the indexing problem. In Section~\ref{sec:sync_decoding}, we provided an algorithm to solve the indexing problem along with synchronization strings with an asymptotic guarantee of $2n\delta$ misdecodings. As explained in Section~\ref{sec:introglobaldecoding}, such a guarantee falls short of giving Theorem~\ref{thm:main}. In this section, we thus provide a variety of more advanced decoding algorithms that provide a better decoding guarantees, in particular achieve a misdecoding fraction which goes to zero as $\eps$ goes to zero. 

We start by pointing out a very useful property of $\eps$-synchronization strings in Section~\ref{sec:SelfMatching}. We define a monotone matching between two strings as a common subsequence of them. We will next show that in a monotone matching between an $\eps$-synchronization string and itself, the number of matches that both correspond to the same element of the string is fairly large. We will refer to this property as $\eps$-self-matching property. We show that one can very formally think of this $\eps$-self-matching property as a robust global guarantee in contrast to the factor-closed strong local requirements of the $\eps$-synchronization property. One advantage of this relaxed notion of $\eps$-self-matching is that one can show that a random string over alphabets polynomially large in $\eps^{-1}$ satisfies this property (Section~\ref{sec:PolynomialConstructionOfEpsSelfMatching}). This leads to a particularly simple generation process for $S$. Finally, showing that this property even holds for approximately $\log n$-wise independent strings directly leads to a deterministic polynomial time algorithm generating such strings as well. 

In Section~\ref{sec:indelErrors}, we propose a decoding algorithm for insdel errors that basically works by finding monotone matchings between the received string and the synchronization string. Using the $\eps$-self-matching property we show that this algorithm guarantees $O\left(n\sqrt\eps\right)$ misdecodings. This algorithm works in time $O(n^2/\sqrt\eps)$ and is exactly what we need to prove our main theorem.

Lastly, in Sections~\ref{sec:delOnly} and \ref{sec:insErrors} we provide two simpler linear time algorithms that solve the indexing problem under the assumptions that the adversary can only delete symbols or only insert new symbols. These algorithms not only guarantee asymptotically optimal $\frac{\eps}{1-\eps}n\delta$ misdecodings but are also error-free. Table~\ref{tbl:indexingSolutions} provides a break down of the decoding schemes presented in this paper, describing the type of error they work under, the number of misdecodings they guarantee, whether they are error-free or streaming, and their decoding complexity.

\shortOnly{Proves of the theorems stated in this section can be found in Appendix~\ref{app:improved_decoding}.}

\subsection{Monotone Matchings and the $\eps$-Self Matching Property}\label{sec:SelfMatching}

Before proceeding to the main results of this section, we start by defining \emph{monotone matchings} which provide a formal way to refer to common substrings of two strings:

\begin{definition}[Monotone Matchings]
A monotone matching between $S$ and $S'$ is a set of pairs of indices like:
$$M=\{(a_1, b_1), \cdots, (a_m, b_m)\}$$
where $a_1 < \cdots < a_m$, $b_1 < \cdots < b_m$, and $S[a_i] = S'[b_i]$.
\end{definition}

We now point out a key property of synchronization strings that will be broadly used in our decoding algorithms. Basically, Theorem~\ref{thm:property} states that two similar subsequences of an $\eps$-synchronization string cannot disagree on many positions.
More formally, let $M = \left\{(a_1, b_1), \cdots, (a_m, b_m)\right\}$ be a monotone matching between $S$ and itself. We call the pair $(a_i, b_i)$ a \emph{good pair} if $a_i = b_i$ and a \emph{bad pair} otherwise. Then:
\begin{theorem}\label{thm:property}
Let $S$ be an $\eps$-synchronization string of size $n$ and 
$M=\left\{(a_1, b_1), \cdots, (a_m, b_m)\right\}$
be a monotone matching of size $m$ from $S$ to itself containing $g$ good pairs and $b$ bad pairs. Then,
\STOConly{$b \le \eps(n - g)$.}
\noSTOC{$$b \le \eps(n - g)$$}
\end{theorem}

\global\def\SynchToSelfMatchingProperty{
\shortOnly{\begin{proof}[of Theorem~\ref{thm:property}]}
\fullOnly{\begin{proof}}
Let $(a'_1, a'_2), \cdots, (a'_{m'}, b'_{m'})$ indicate the set of bad pairs in $M$ indexed as 
$a'_1<\cdots<a'_{m'}$ and $b'_1<\cdots<b'_{m'}$.
Without loss of generality, assume that $a'_1 < b'_1$. Let $k_1$ be the largest integer such that $a'_{k_1} < b'_1$. Then, the pairs $(a'_1, a'_2), \cdots, (a'_{k_1}, b'_{k_1})$ form a common substring of size $k_1$ between $T_1 = S[a'_1, b'_1)$ and $T'_1 = S[b'_1, b'_{k_1}]$. Now, the synchronization string guarantee implies that:
\noSTOC{\begin{eqnarray*}
k_1 &\le& LCS(T_1, T'_1)\\
&\le& \frac{|T_1| + |T'_1| - ED\left(T_1, T'_1\right)}{2}\\
&\le&\frac{\eps(|T_1| + |T'_1|)}{2}
\end{eqnarray*}}
\STOConly{
$$k_1 \le LCS(T_1, T'_1)
\le \frac{|T_1| + |T'_1| - ED\left(T_1, T'_1\right)}{2}
\le\frac{\eps(|T_1| + |T'_1|)}{2}
$$}

Note that the monotonicity of the matching guarantees that there are no good matches occurring on indices covered by $T_1$ and $T'_1$, i.e., $a'_1, \cdots, b'_{k_1}$. One can repeat very same argument for the remaining bad matches to rule out bad matches $(a'_{k_1+1}, b'_{k_1+1}), \cdots, (a'_{k_1+k_2}, b'_{k_1+k_2})$ for some $k_2$ having the following inequality guaranteed:

\begin{eqnarray}
k_2 \le \frac{\eps(|T_2| + |T'_2|)}{2}\label{eq:piecewise}
\end{eqnarray}

where 
$$\Bigg\{
\begin{tabular}{cc}
$T_2 = [a'_{k_1+1}, b'_{k_1+1})$ and $T'_2 = [b'_{k_1+1}, b'_{k_1+k_2}]$ & $a'_{k_1+1} < b'_{k_1+1}$ \\[2mm]
$T_2 = [b'_{k_1+1}, a'_{k_1+1})$ and $T'_2 = [a'_{k_1+1}, a'_{k_1+k_2}]$ & $a'_{k_1+1} > b'_{k_1+1}$
\end{tabular}$$
For a pictorial representation see Figure~\ref{fig:property}.

\STOConly{
\begin{figure}
\centering
\begin{subfigure}{.24\textwidth}
  \centering
  \includegraphics[width=1\linewidth]{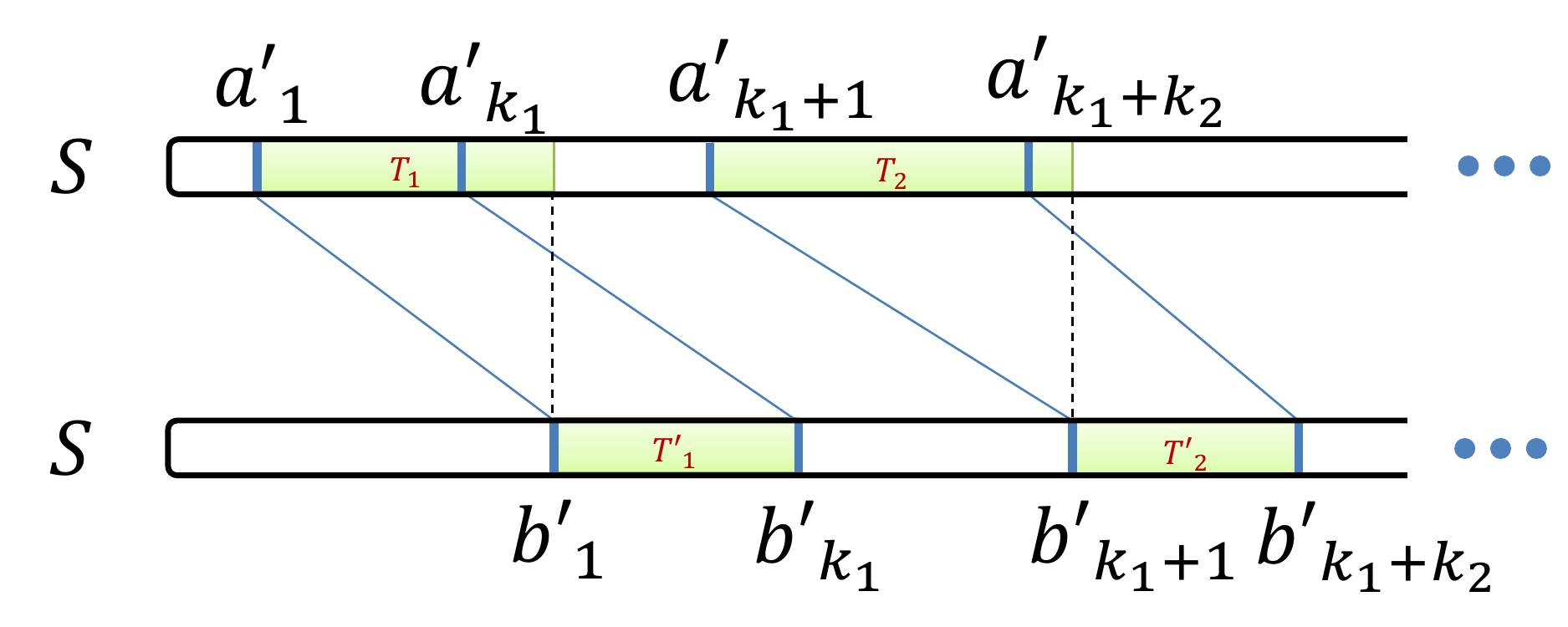}
  \caption{$a'_{k_1+1} < b'_{k_1+1}$}
  \label{fig:sub1}
\end{subfigure}%
\begin{subfigure}{.24\textwidth}
  \centering
  \includegraphics[width=1\linewidth]{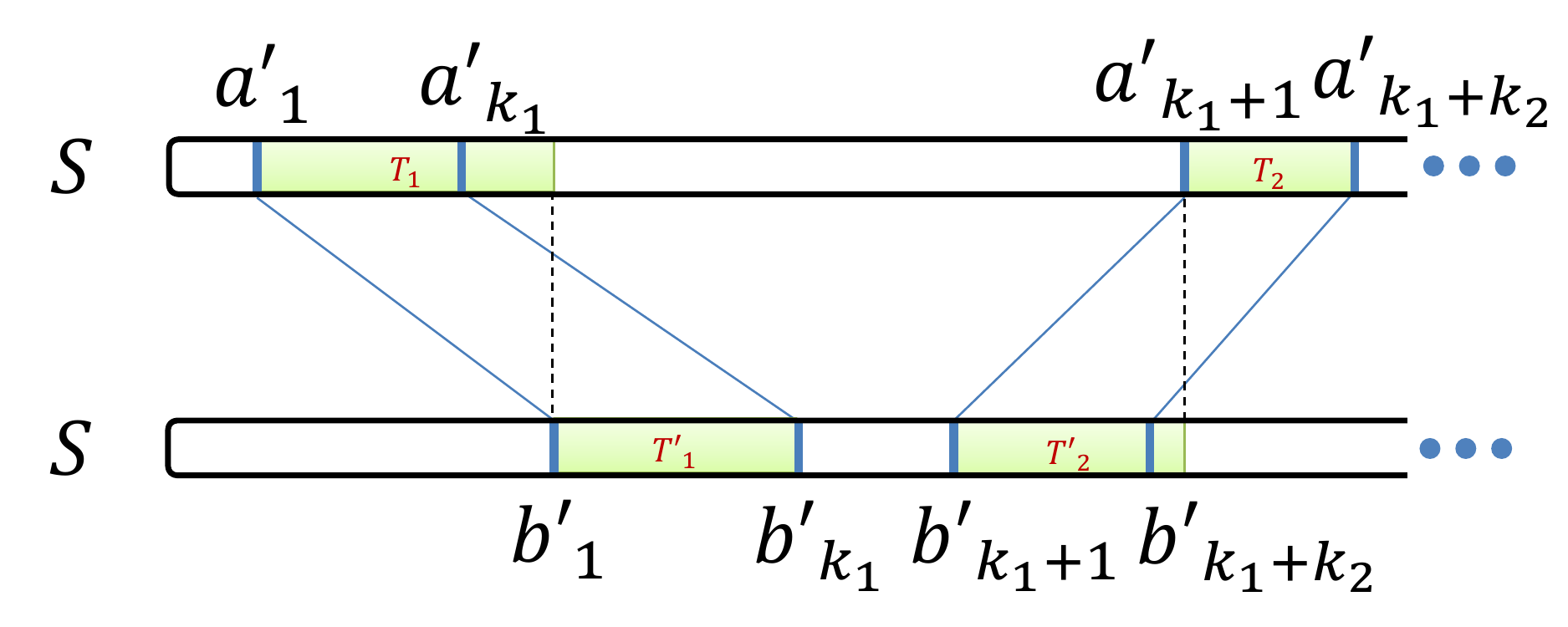}
  \caption{$a'_{k_1+1} > b'_{k_1+1}$}
  \label{fig:sub2}
\end{subfigure}
\caption{Pictorial representation of $T_2$ and $T'_2$}
\label{fig:property}
\end{figure}}

\noSTOC{
\begin{figure}
\centering
\begin{subfigure}{.5\textwidth}
  \centering
  \includegraphics[width=1\linewidth]{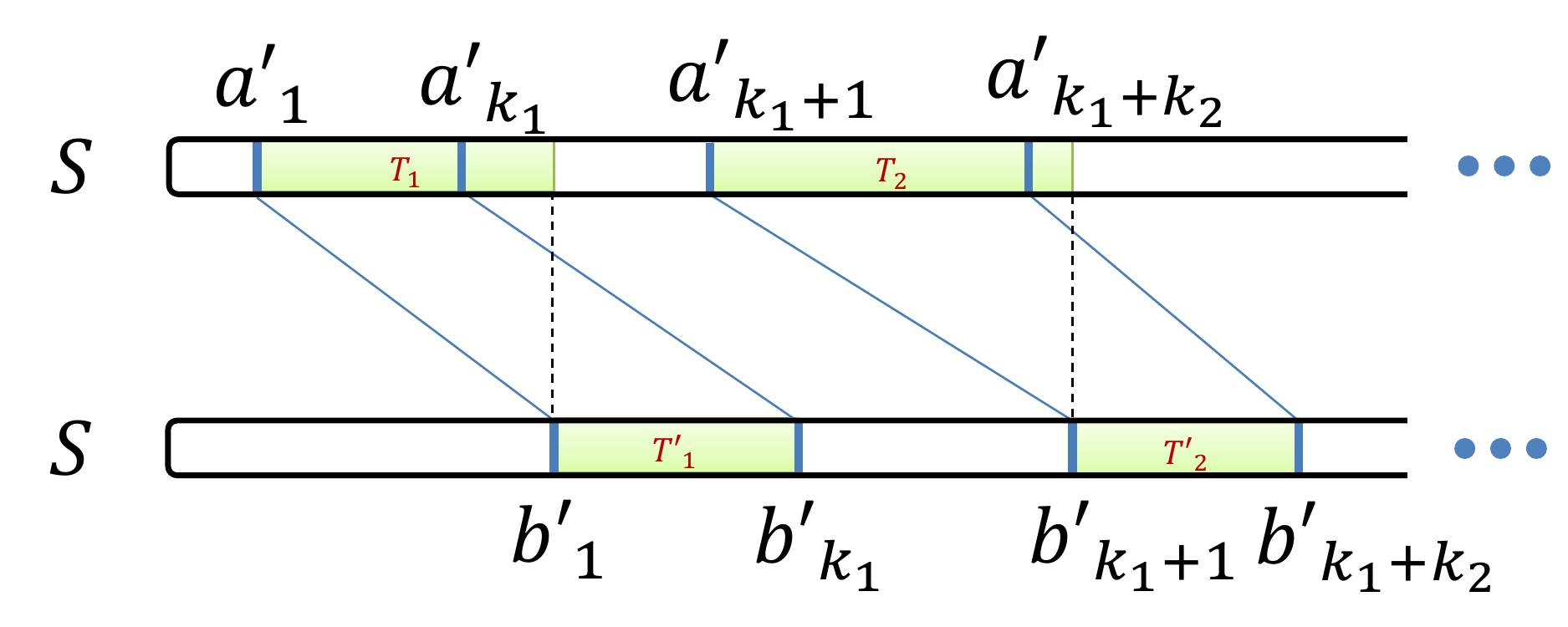}
  \caption{$a'_{k_1+1} < b'_{k_1+1}$}
  \label{fig:sub1}
\end{subfigure}%
\begin{subfigure}{.5\textwidth}
  \centering
  \includegraphics[width=1\linewidth]{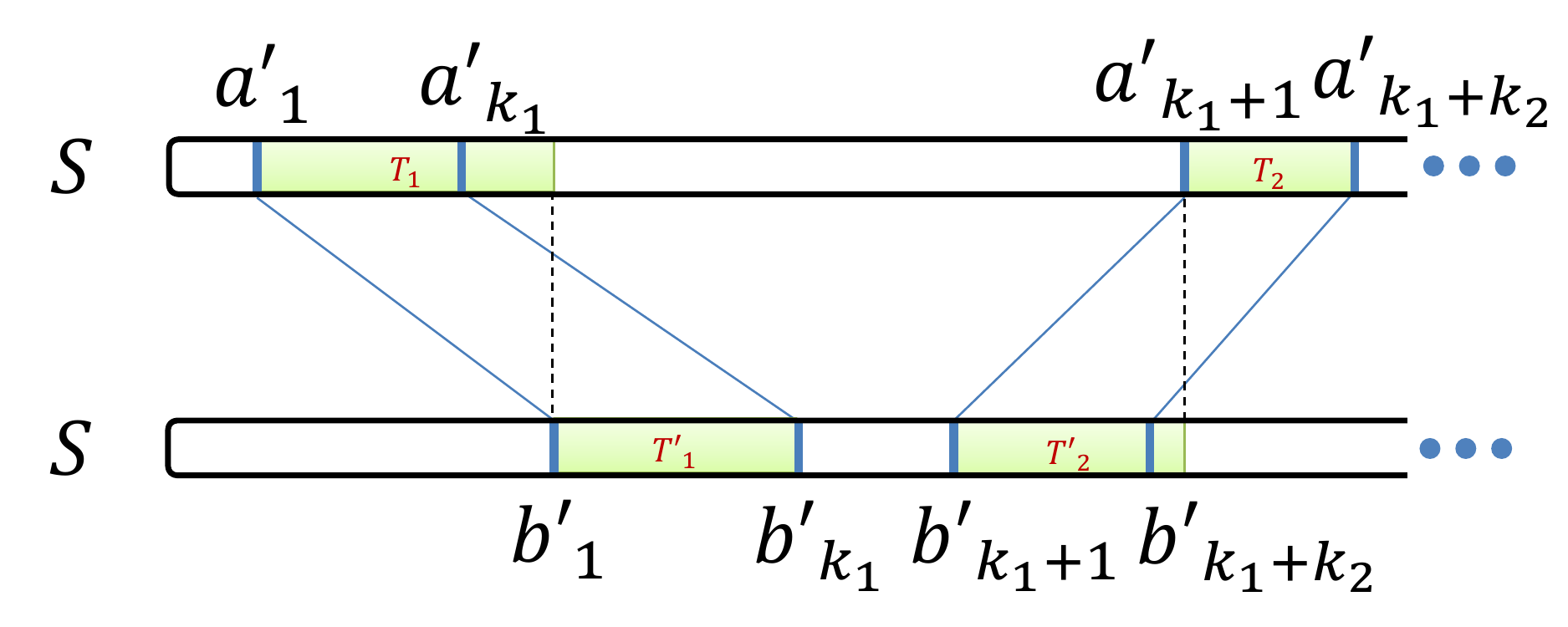}
  \caption{$a'_{k_1+1} > b'_{k_1+1}$}
  \label{fig:sub2}
\end{subfigure}
\caption{Pictorial representation of $T_2$ and $T'_2$}
\label{fig:property}
\end{figure}}

Continuing the same procedure, one can find $k_1, \cdots, k_l$, $T_1, \cdots, T_l$, and $T'_1, \cdots, T'_l$ for some $l$. Summing up all inequalities of form \eqref{eq:piecewise}, we will have:
\begin{eqnarray}
\sum_{i=1}^l k_i \le \frac{\eps}{2}\cdot\left(\sum_{i=1}^l |T_i| + \sum_{i=1}^l |T'_i|\right)\label{eq:summation}
\end{eqnarray}

Note that $\sum_{i=1}^l k_i = u$ and $T_i$s are mutually exclusive and contain no indices where a good pair occurs at. Same holds for $T'_i$s. Hence, $\sum_{i=1}^l |T_i| \le n-g$ and $\sum_{i=1}^l |T'_i| \le n-g$. All these along with \eqref{eq:summation} give that:
\begin{eqnarray*}
u \le \frac{\eps}{2}\cdot2\left(n-g\right) = \eps(n-g) \Rightarrow n - g - b \ge (1-\eps)(n-g) \Rightarrow b \le \eps(n-g)
\end{eqnarray*}
\end{proof}
}
\fullOnly{\SynchToSelfMatchingProperty}

We define the $\eps$-self-matching property as follows:
\begin{definition}[$\eps$-self-matching property]
String $S$ satisfies $\eps$-self-matching property if any monotone matching between $S$ and itself contains less than $\eps|S|$ bad pairs.
\end{definition}

Note that $\eps$-synchronization property concerns all substrings of a string while the $\eps$-self-matching property only concerns the string itself. Granted that, we now show that $\eps$-synchronization property and satisfying $\eps$-self-matching property on all substrings are equivalent up to a factor of two:

\begin{theorem}\label{lemma:simplifiedProperty}
$\eps$-synchronization and $\eps$-self matching properties are related as follows:
\begin{itemize}
\item [a)] If $S$ is an $\eps$-synchronization string, then all substrings of $S$ satisfy $\eps$-self-matching property.
\item [b)] If all substrings of string $S$ satisfy the $\frac{\eps}{2}$-self-matching property, then $S$ is $\eps$-synchronization string.
\end{itemize}
\end{theorem}
\shortOnly{\begin{proof}[of Theorem~\ref{lemma:simplifiedProperty} (a)]}
\fullOnly{\begin{proof}[Proof of Theorem~\ref{lemma:simplifiedProperty} (a)]}
This part is a straightforward consequence of Theorem~\ref{thm:property}.\end{proof}
\shortOnly{\begin{proof}[of Theorem~\ref{lemma:simplifiedProperty} (b)]}
\fullOnly{\begin{proof}[Proof of Theorem~\ref{lemma:simplifiedProperty} (b)]}
Assume by contradiction that there are $i < j < k$ such that $ED(S[i, j), S[j, k)) \le (1-\eps) (k-i)$. Then, $LCS(S[i, j), S[j, k)) \ge \frac{k - i - (1-\eps) (k-i)}{2} = \frac{\eps}{2} (k-i)$. The corresponding pairs of such longest common substring form a monotone matching of size $\frac{\eps}{2} (k-i)$ which contradicts $\frac{\eps}{2}$-self-matching property of $S$.
\end{proof}

As a matter of fact, the decoding algorithms we will propose for $\eps$-synchronization strings in Sections~\ref{sec:indelErrors}, \ref{sec:delOnly}, and~\ref{sec:insErrors} only make use of the $\eps$-self-matching property of the $\eps$-synchronization string. 

We now proceed to the definition of $\eps$-bad-indices which will enable us to show that $\eps$-self matching property, as opposed to the $\eps$-synchronization property, is robust against local changes.

\begin{definition}[$\eps$-bad-index]
We call index $k$ of string $S$ an $\eps$-bad-index if there exists a factor $S[i, j]$ of $S$ with $i\le k\le j$ where $S[i, j]$ does not satisfy the $\eps$-self-matching property. In this case, we also say that index $k$ blames interval $[i, j]$.
\end{definition}

Using the notion of $\eps$-bad indices, we now present Lemma~\ref{lemma:EpsGammaSelfMatching}. This lemma suggests that a string containing limited fraction of $\eps$-bad indices would still be an $\eps'$-self matching string for some $\eps' > \eps$. An important consequence of this result is that if one changes a limited number of elements in a given $\eps$-self matching string, the self matching property will be essentially preserved to a lesser extent. Note that $\eps$-synchronization property do not satisfy any such robustness quality.

\begin{lemma}\label{lemma:EpsGammaSelfMatching}
If the fraction of $\eps$-bad indices in string $S$ is less than $\gamma$, then $S$ satisfies $(\eps+2\gamma)$-self matching property. 
\end{lemma}
\begin{proof}
Consider a matching from $S$ to itself. The number of bad matches whose both ends refer to non-$\eps$-bad indices of $S$ is at most $|S|(1-\gamma)\eps$ by definition. Further, each $\eps$-bad index can appear at most once in each end of bad pairs. Therefore, the number of bad pairs in $S$ can be at most:
$$|S|(1-\gamma)\eps + 2|S|\gamma \le |S|(\eps+2\gamma)$$
which, by definition, implies that $S$ satisfies the $(\eps+2\gamma)$-self-matching property.
\end{proof}

On the other hand, in the following theorem, we will show that within a given $\eps$-self matching string, there can be a limited number of $\eps'$-bad indices for sufficiently large $\eps' > \eps$.

\begin{lemma}\label{lem:EpsBadIndsInEpsBadStrings}
Let $S$ be an $\eps$-self matching string of length $n$. Then, for any $3\eps<\eps'<1$, at most $\frac{3n\eps}{\eps'}$ many indices of $S$ can be $\eps'$-bad.
\end{lemma}

\global\def\ProofOfEpsBadIndsInEpsBadStrings{
\fullOnly{\begin{proof}}
\shortOnly{\begin{proof}[of Lemma~\ref{lem:EpsBadIndsInEpsBadStrings}]}
Let $s_1, s_2, \cdots, s_k$ be bad indices of $S$ and $\eps'$-bad index $s_i$ blame substring $S[a_i, b_i)$. As intervals $S[a_i, b_i)$ are supposed to be bad, there has to be a $\eps'$-self matching within each $S[a_i, b_i)$ like $M_i$ for which $|M_i| \ge \eps' \cdot |[a_i, b_i)|$. We claim that one can choose a subset of $[1..k]$ like $I$ for which 
\begin{itemize}
\item Corresponding intervals to the indices in $I$ are mutually exclusive. In other words, for any $i, j \in I$ where $i \not = j$, $[A_i, b_i) \cap [a_j, b_j) = \emptyset$.
\item $\sum_{i\in I} \left|[A_i, b_i)\right| \ge \frac{k}{3}$.
\end{itemize}

If such $I$ exists, one can take $\bigcup_{i\in I} M_i$ as a self matching in $S$ whose size is larger than $\frac{k\eps'}{3}$. As $S$ is an $\eps$-self matching string, 
$$\frac{k\eps'}{3} \le n\eps \Rightarrow k \le \frac{3n\eps}{\eps'}$$
which finishes the proof. The only remaining piece is proving the claim. Note that any index in $\bigcup_{i\in I}[a_i, b_i)$ is a $\eps'$-bad index as they by definition belong to an interval with a $\eps'$-self matching. Therefore, $\left|\bigcup_{i\in I}[a_i, b_i)\right| = k$. In order to find set $I$, we greedily choose the largest substring $[a_i, b_i)$, put its corresponding index into $I$ and then remove any interval intersecting $[a_i, b_i)$. We continue repeating this procedure until all substrings are removed. The set $I$ obtained by this procedure clearly satisfies the first claimed property. Moreover, note that if $l_i = \left|[a_i, b_i)\right|$, any interval intersecting $[a_i, b_i)$ falls into $[a_i-l_i, b_i+l_i)$ which is an interval of length $3l_i$. This certifies the second property and finishes the proof.
\end{proof}
}
\fullOnly{\ProofOfEpsBadIndsInEpsBadStrings}

As the final remark on the $\eps$-self matching property and its relation with the more strict $\eps$-synchronization property, we show that using the minimum RSD decoder for indexing together with an $\eps$-self matching string leads to guarantees on the misdecoding performance which are only slightly weaker than the guarantee obtained by $\eps$-synchronization strings. In order to do so, we first show that the $(1-\eps)$ RSD distance property of prefixes holds for any non-$\eps$-bad index in any arbitrary string in Theorem~\ref{thm:RSDforEpsSelfBadIndxs}. 
Then, using Theorem~\ref{thm:RSDforEpsSelfBadIndxs} and Lemma~\ref{lem:EpsBadIndsInEpsBadStrings}, we upper-bound the number of misdecodings that may happen using a minimum RSD decoder along with an $\eps$-self matching string in Theorem~\ref{thm:MinRSDforEpsSelfMatchingFree}.

\begin{theorem}\label{thm:RSDforEpsSelfBadIndxs}
Let $S$ be an arbitrary string of length $n$ and $1\le i\le n$ be such that $i$'th index of $S$ is not an $\eps$-bad index. Then, for any $j\not = i$, 
$RSD(S[1, i], S[1, j]) > 1-\eps$.
\end{theorem}

\begin{proof}
Without loss of generality assume that $j < i$. Consider the interval $[2j - i+1, i]$. As $i$ is not an $\eps$-bad index, there is no self matching of size $2\eps(i-j)$ within $[2j - i, i]$. In particular, the edit distance of $S[2j - i+1, j]$ and $[j+1, i]$ has to be larger than $(1-\eps)\cdot 2(i-j)$ which equivalently means $RSD(S[1, i], S[1, j]) > 1-\eps$. Note that if $2j - i+1 < 0$ the proof goes through by simply replacing $2j - i+1$ with zero.
\end{proof}

\begin{theorem}\label{thm:MinRSDforEpsSelfMatchingFree}
Using any $\eps$-self matching string along with minimum RSD algorithm, one can solve the $(n, \delta)$-indexing problem with a guarantee of $n(4\delta + 6\eps)$ misdecodings.
\end{theorem}

\begin{proof}
Note that applying Lemma~\ref{lem:EpsBadIndsInEpsBadStrings} for $\eps'$ gives that there are at most $\frac{3n\eps}{\eps'}$ indices in $S$ that are $\eps'$-bad. Further, using Theorems~\ref{thm:MinRSDMisDec}~and~\ref{thm:RSDforEpsSelfBadIndxs}, at most $\frac{2n\delta}{1-\eps'}$ many of the other indices might be decoded incorrectly upon their arrivals. Therefore, this solution for the $(n, \delta)$-indexing problem can contain at most $n\left(\frac{3\eps}{\eps'} + \frac{2\delta}{1-\eps'}\right)$ many incorrectly decoded indices. Setting $\eps' = \frac{3\eps}{3\eps+2\delta}$ gives an upper bound of $n(4\delta + 6\eps)$ on the number of misdecodings.
\end{proof}

\subsection{Efficient Polynomial Construction of $\eps$-self matching strings}\label{sec:PolynomialConstructionOfEpsSelfMatching}

In this section, we will use Lemma~\ref{lemma:EpsGammaSelfMatching} to show that there is a polynomial deterministic construction of a string of length $n$ with the $\eps$-self-matching property, which can then for example be used to obtain a deterministic code construction. We start by showing that even random strings satisfy the $\eps$-selfmatching property for an $\eps$ polynomial in the alphabet size:

\begin{theorem}\label{thm:EpsSelfMatchingRandomConstruction}
A random string on an alphabet of size $O(\eps^{-3})$ satisfies $\eps$-selfmatching property with a constant probability.
\end{theorem}

\global\def\ProofOfEpsSelfMatchingRandomConstruction{
\fullOnly{\begin{proof}}
\shortOnly{\begin{proof}[of Theorem~\ref{thm:EpsSelfMatchingRandomConstruction}]}
Let $S$ be a random string on alphabet $\Sigma$ of size $|\Sigma|=\eps^{-3}$. We are going to find the expected number of $\eps$-bad indices in $S$. We first count the expected number of $\eps$-bad indices that blame intervals of length $\frac{2}{\eps}$ or smaller. If index $k$ blames interval $S[i, j]$ where $j-i < 2\eps^{-1}$, there has to be two identical symbols appearing in $S[i, j]$ which gives that there are two identical elements in $4\eps^{-1}$ neighborhood of $S$. Therefore, the probability of index $k$ being $\eps$-bad blaming $S[i,j]$ for $j-i<2\eps^{-1}$ can be upper-bounded by
${4\eps^{-1} \choose 2}\frac{1}{|\Sigma|} \le 8\eps$. Thus, the expected fraction of $\eps$-bad indices that blame intervals of length $\frac{2}{\eps}$ or smaller is less than $8\eps$.

We now proceed to finding the expected fraction of $\eps$-bad indices in $S$ blaming intervals of length $2\eps^{-1}$ or more.
Since every interval of length $l$ which does not satisfy $\eps$-self-matching property causes at most $l$ $\eps$-bad indices, we get that the expected fraction of such indices, i.e., $\gamma'$,  is at most:

\begin{eqnarray}
\mathbb{E}[\gamma']&=&\frac{1}{n}\sum_{i=2\eps^{-1}}^n \sum_{l=1}^n l\cdot\Pr[S[i, i+l) \text{ does not satisfy $\eps$-self-matching property}]\nonumber\\
&=& \sum_{l=2\eps^{-1}}^n l\cdot\Pr[S[i, i+l) \text{ does not satisfy $\eps$-self-matching property}]\nonumber\\
&\le& \sum_{l=2\eps^{-1}}^n l {l\choose l\eps}^2 \frac{1}{|\Sigma|^{l\eps}} \label{eq:independentStep}
\end{eqnarray}
Last inequality holds because the number of possible matchings is at most ${l\choose l\eps}^2$. Further, fixing the matching edges, the probability of the elements corresponding to pair $(a, b)$ of the matching being identical is independent from all pairs $(a', b')$ where $a'<a$ and $b'<b$. Hence, the probability of the set of pairs being a matching between random string $S$ and itself is $\frac{1}{|\Sigma|^{l\eps}}$. Then,
\noSTOC{\begin{eqnarray}
\mathbb{E}[\gamma']&\le& \sum_{l=2\eps^{-1}}^n l \left(\frac{le}{l\eps}\right)^{2l\eps} \frac{1}{|\Sigma|^{l\eps}}\nonumber\\
&\le& \sum_{l=2\eps^{-1}}^n l \left(\frac{e}{\eps\sqrt{|\Sigma|}}\right)^{2\eps l}\nonumber\\
&\le& \sum_{l=2\eps^{-1}}^\infty l \left[\left(\frac{e}{\eps\sqrt{|\Sigma|}}\right)^{2\eps}\right]^{l}\nonumber
\end{eqnarray}
}
\STOConly{\begin{eqnarray}
\mathbb{E}[\gamma']&\le& \sum_{l=2\eps^{-1}}^n l \left(\frac{le}{l\eps}\right)^{2l\eps} \frac{1}{|\Sigma|^{l\eps}}
\le \sum_{l=2\eps^{-1}}^n l \left(\frac{e}{\eps\sqrt{|\Sigma|}}\right)^{2\eps l}\nonumber\\
&\le& \sum_{l=2\eps^{-1}}^\infty l \left[\left(\frac{e}{\eps\sqrt{|\Sigma|}}\right)^{2\eps}\right]^{l}\nonumber
\end{eqnarray}
}
Note that series $\sum_{l=2\eps^{-1}}^\infty l x^l = \frac{2\eps^{-1}x^{2\eps^{-1}} - (2\eps^{-1}-1)x^{2\eps^{-1}+1}}{(1-x)^2}$ for $|x| < 1$. Therefore, for $0 < x < \frac{1}{2}$, $\sum_{l=l_0}^\infty l x^l < 8\eps^{-1}x^{2\eps^{-1}}$. So,
\begin{eqnarray*}
\mathbb{E}[\gamma'] &\le& 8\eps^{-1}\left(\frac{e}{2\eps\sqrt{|\Sigma|}}\right)^{4\eps\eps^{-1}} = \frac{e^4}{2}\eps^{-5}\frac{1}{|\Sigma|^2}\le \frac{e^4}{2}\eps
\end{eqnarray*}
Using Lemma~\ref{lemma:EpsGammaSelfMatching}, this random structure has to satisfy $(\eps+2\gamma)$-self-matching property where
$$\mathbb{E}[\eps+2\gamma] = \eps + 16\eps + e^4\eps = O(\eps)$$
Therefore, using Markov inequality, a randomly generated string over alphabet $O(\eps^{-3})$ satisfies $\eps$-matching property with constant probability. The constant probability can be as high as one wishes by applying higher constant factor in alphabet size.
\end{proof}
}

\fullOnly{\ProofOfEpsSelfMatchingRandomConstruction}

As the next step, we prove a similar claim for strings of length $n$ whose symbols are chosen from an $\Theta\left(\frac{\log n}{\log (1/\eps)}\right)$-wise independent~\cite{naor1993small} distribution over a larger, yet still $\eps^{-O(1)}$ size, alphabet. This is the key step in allowing for a derandomization using the small sample spaces of Naor and Naor~\cite{naor1993small}. The proof of Theorem~\ref{thm:LimitedIndependenceRelaxedConstruction} follows a similar strategy as was used in~\cite{HaeuplerSICOMP13p2155} to derandomize the constructive Lov\'{a}sz local lemma. In particular the crucial idea, given by Claim~\ref{claim:submatch}, is to show that for any large obstruction there has to exist a smaller yet not too small obstruction. This allows one to prove that in the absence of any small and medium size obstructions no large obstructions exist either. 

\begin{theorem}\label{thm:LimitedIndependenceRelaxedConstruction}
A $\frac{c\log n}{\log (1/\eps)}$-wise independent random string of size $n$ on an alphabet of size $O(\eps^{-6})$ satisfies $\eps$-matching property with a non-zero constant probability. $c$ is a sufficiently large constant.
\end{theorem}

\global\def\ProofOfLimitedIndependenceRelaxedConstruction{
\fullOnly{\begin{proof}}
\shortOnly{\begin{proof}[of Theorem~\ref{thm:LimitedIndependenceRelaxedConstruction}]}
Let $S$ be a pseudo-random string of length $n$ with $\frac{c\log n}{\log (1/\eps)}$-wise independent symbols. Then, Step~\eqref{eq:independentStep} is invalid as the proposed upper-bound does not work for $l > \frac{c\log n}{\eps\log(1/\eps)}$. To bound the probability of intervals of size $\Omega\left(\frac{c\log n}{\eps\log(1/\eps)}\right)$ not satisfying $\eps$-self matching property, we claim that:
\begin{claim}\label{claim:submatch}
Any string of size $l >100 m$ which contains an $\eps$-self-matching contains two sub-intervals $I_1$ and $I_2$ of size $m$ where there is a matching of size $0.99\frac{m\eps}{2}$ between $I_1$ and $I_2$.
\end{claim}
Using Claim~\ref{claim:submatch}, one can conclude that 
any string of size $l > 100\frac{c\log n}{\eps\log(1/\eps)}$ which contains an $\eps$-self-matching contains two sub-intervals $I_1$ and $I_2$ of size $\frac{c\log n}{\eps\log(1/\eps)}$ where there is a matching of size $\frac{c\log n}{2\log (1/\eps)}$ between $I_1$ and $I_2$.
Then, Step~\eqref{eq:independentStep} can be revised by upper-bounding the probability of a long interval having an $\eps$-self-matching by a union bound over the probability of pairs of its subintervals having a dense matching. Namely, for $l > 100\frac{c\log n}{\eps\log(1/\eps)}$, let us denote the event of $S[i, i+l)$ containing a $\eps$-self-matching by $A_{i, l}$. Then,
\begin{eqnarray*}
\Pr[A_{i, l}] &\le& \Pr\left[S\text{ contains $I_1$, $I_2: |I_i| = \frac{c\log n}{\eps\log(1/\eps)}$ and $LCS(I_1, I_2) \ge 0.99\frac{c\log n}{2\log(1/\eps)}$}\right]
\\
&\le& n^2 {\eps^{-1}c\log n/\log(1/\eps) \choose 0.99c\log n/2\log(1/\eps)}^2\left(\frac{1}{|\Sigma|}\right)^\frac{c\log n}{2\log(1/\eps)}
\\
&\le& n^2 \left(2.04e\eps^{-1}\right)^\frac{2\times0.99c\log n}{2\log(1/\eps)}\eps^\frac{6c\log n}{2\log(1/\eps)}
\\ 
&=& n^2 \left(2.04e\right)^\frac{0.99c\log n}{\log(1/\eps)} \eps^\frac{4.02c\log n}{2\log(1/\eps)}
\\ 
&=& n^{2 + \frac{c\ln(2.04e)}{\log(1/\eps)}  -2.01c} < n^{2 -c/4} = O\left(n^{-c'}\right)
\end{eqnarray*}
where first inequality follows from the fact there can be at most $n^2$ pairs of intervals of size $\frac{c\log n}{\eps\log(1/\eps)}$ in $S$ and the number of all possible matchings of size $\frac{c\log n}{\log (1/\eps)}$ between them is at most ${\eps^{-1}c\log n/\log(1/\eps) \choose c\log n/2\log(1/\eps)}^2$. Further, for small enough $\eps$, constant $c'$ can be as large as one desires by setting constant $c$ large enough. Thus, Step~\eqref{eq:independentStep} can be revised as:
\begin{eqnarray*}
\mathbb{E}[\gamma'] &\le& \sum_{l=\eps^{-1}}^{100c\log n/(\eps\log(1/\eps))} l\cdot\left[\left(\frac{e}{\eps\sqrt{|\Sigma|}}\right)^{2\eps}\right]^{l}
+\sum_{l=100c\log n/(\eps\log(1/\eps))}^{n} l  \Pr[A_{i, l}]
\\
&\le& \sum_{l=\eps^{-1}}^{\infty} l\cdot\left[\left(\frac{e}{\eps\sqrt{|\Sigma|}}\right)^{2\eps}\right]^{l}
+n^2\cdot O(n^{-c'})
\le O(\eps) + O(n^{2-c'})
\end{eqnarray*}
For an appropriately chosen $c$, $2-c' < 0$; hence, the later term vanishes as $n$ grows. Therefore, the conclusion $\mathbb{E}[\gamma] \le O(\eps)$ holds for the limited $\frac{\log n}{\log(1/\eps)}$-wise independent string as well. 
\end{proof}

\fullOnly{\begin{proof}[Proof of Claim~\ref{claim:submatch}]}
\shortOnly{\begin{proof}[of Claim~\ref{claim:submatch}]}
Let $M$ be a self-matching of size $l\eps$ or more between $S$ and itself containing only bad edges. We chop $S$ into $\frac{l}{m}$ intervals of size $m$. On the one hand, the size of $M$ is greater than $l\eps$ and on the other hand, we know that the size of $M$ is exactly $\sum_{i, j}|E_{i,j}|$ where $E_{i,j}$ denotes the number of edges between interval $i$ and $j$. Thus:
$$l\eps \le \sum_{i, j}|E_{i,j}| \Rightarrow  \frac{\eps}{2} \le \frac{\sum_{i, j} |E_{i, j}|/m}{{2l}/{m}}$$
Note that $\frac{|E_{i, j}|}{m}$ represents the density of edges between interval $i$ and interval $j$. 
Further, Since $M$ is monotone, there are at most $\frac{2l}{m}$ intervals for which $|E_{i, j}|\not = 0$ and subsequently $\frac{|E_{i, j}|}{m} \not = 0$. Hence, on the right hand side we have the average of  $\frac{2l}{m}$ many non-zero terms which is greater than $\eps/2$. So, there has to be some $i'$ and $j'$ for which:
$$\frac{\eps}{2} \le \frac{|E_{i', j'}|}{m} \Rightarrow \frac{m\eps}{2}\le |E_{i',j'}|$$
To analyze more accurately, if $l$ is not divisible by $m$, we simply throw out up to $m$ last elements of the string. This may decrease $\eps$ by $\frac{m}{l} < \frac{\eps}{100}$. 
\end{proof}
}

\fullOnly{\ProofOfLimitedIndependenceRelaxedConstruction}

Note that using the polynomial sample spaces of~\cite{naor1993small} Theorem~\ref{thm:LimitedIndependenceRelaxedConstruction} directly leads to a deterministic algorithm for finding a string of size $n$ with $\eps$-self-matching property. For this one simply checks all possible points in the sample space of the $\frac{c\log n}{\log(1/\eps)}$-wise independent strings and finds a string $S$ with $\gamma_S \le \mathbb{E}[\gamma] = O(\eps)$. In other words, using brute-force, one can find a string satisfying $O(\eps)$-self-matching property in $O\left(|\Sigma|^{\frac{c\log n}{\log(1/\eps)}}\right) = n^{O(1)}$. 

\begin{theorem}\label{thm:detepsselfconstruction}
There is a deterministic algorithm running in $n^{O(1)}$ that finds a string of length $n$ satisfying $\eps$-self-matching property over an alphabet of size $O(\eps^{-6})$.
\end{theorem}

\subsection{Insdel Errors}\label{sec:indelErrors}
Now, we provide an alternative indexing algorithm to be used along with $\eps$-synchronization strings. 
Throughout the following sections, we let $\eps$-synchronization string $S$ be sent as the synchronization string in an instance of $(n, \delta)$-indexing problem and string $S'$ be received at the receiving end being affected by up to $n\delta$ insertions or deletions.
Furthermore, let $d_i$ symbols be inserted into the communication and $d_r$ symbols be deleted from it.

The algorithm works as follows. On the first round, the algorithm finds the longest common subsequence between $S$ and $S'$. Note that this common subsequence corresponds to a monotone matching $M_1$ between $S$ and $S'$. On the next round, the algorithm finds the longest common subsequence between $S$ and the subsequence of unmatched elements of $S'$ (those that have not appeared in $M_1$). This common subsequence corresponds to a monotone matching between $S$ and the elements of $S'$ that do not appear in $M_1$. The algorithm repeats this procedure $\frac{1}{\beta}$ times to obtain $M_1, \cdots, M_{1/\beta}$ where $\beta$ is a parameter that we will fix later.

In the output of this algorithm, $S'[t_i]$ is decoded as $S[i]$ if and only if $S[i]$ is only matched to $S'[t_i]$ in all $M_1, \cdots, M_{1/\beta}$.  Note that the longest common subsequence of two strings of length $O(n)$ can be found in $O(n^2)$ using dynamic programming. Therefore, the whole algorithm runs in $O\left(n^2/\beta\right)$.

Now we proceed to analyzing the performance of the algorithm by bounding the number of misdecodings. 

\begin{theorem}
This decoding algorithm guarantees a maximum misdecoding count of $(n+d_i-d_r)\beta + \frac{\eps}{\beta}n$. More specifically, for 
$\beta = \sqrt{\eps}$, the number misdecodings will be $O\left(n\sqrt{\eps}\right)$ and running time will be $O\left(n^2/\sqrt{\eps}\right)$.
\end{theorem}
\begin{proof}
First, we claim that at most $(n+d_i-d_r)\beta$ many of the symbols that have been successfully transmitted are not matched in any of $M_1, \cdots, M_{1/\beta}$.
Assume by contradiction that more than $(n+d_i-d_r)\beta$ of the symbols that pass through the channel successfully are not matched in any of $M_1, \cdots, M_{1/\beta}$. Then, there exists a monotone matching of size greater than $(n+d_i-d_r)\beta$ between the unmatched elements of $S'$ and $S$ after $\frac{1}{\beta}$ rounds of finding longest common substrings. Hence, size of any of $M_i$s is at least $(n+d_i-d_r)\beta$. So, the summation of their sizes exceeds $(n+d_i-d_r)\beta \times \frac{1}{\beta} = n+d_i-d_r = |S'|$ which brings us to a contradiction.

Furthermore, as a result of Theorem~\ref{lemma:simplifiedProperty}, any of $M_i$s contain at most $\eps n$ many incorrectly matched elements. Hence, at least $\frac{\eps}{\beta}n$ many of the matched symbols are matched to incorrect index.

Hence, the total number of misdecodings can be bounded by $(n+d_i-d_r)\beta + \frac{\eps}{\beta}n$.
\end{proof}

\noSTOC{
\subsection{Deletion Errors Only}\label{sec:delOnly}
We now introduce a very simple linear time streaming algorithm that decodes a received synchronization string of length $n$ which can be affected by up to $n\delta$ many deletions. Our scheme is guaranteed to have less than $\frac{\eps}{1-\eps}\cdot n\delta$ misdecodings.

Before proceeding to the algorithm description, let $d_r$ denote the number of symbols removed by adversary. As adversary is restricted to symbol deletion, each symbol received at the receiver corresponds to a symbol sent by the sender. Hence, there exists a monotone matching of size $|S'| = n' = n-d_r$ like
$M = \{(t_1, 1), (t_2, 2), \cdots, (t_{n-d_r}, n-d_r)\}$
between $S$ and $S'$ which matches each of the received symbols to their actual indices.

Our simple streaming algorithm greedily matches $S$ to the left-most possible subsequence of $S$. To put it another words, the algorithm matches $S'[1]$ to $S[t'_1]$ where $S[t'_1] = S'[1]$ and $t'_1$ is as small as possible, then matches $S'[2]$ to the smallest $t'_2>t'_1$ where $S[t'_2] = S'[2]$ and construct the whole matching $M'$ by repeating this procedure. Note that as there is a matching of size $|S'|$ between $S$ and $S'$, the size of $M'$ will be $|S'|$ too.

This algorithm clearly works in a streaming manner and runs in linear time. To analyze the performance, we basically make use of the fact that $M$ and $M'$ are both monotone matchings of size $|S'|$ between $S$ and $S'$. Therefore, 
$\bar{M} = \{(t_1, t'_1), (t_2, t'_2), \cdots, (t_{n-d_r}, t'_{n-d_r})\}$
is a monotone matching between $S$ and itself.
Note that if $t_i \not = t'_i$, then the algorithm has decoded the index $t_i$ incorrectly. Let $p$ be the number of indices $i$ where $t_i \not = t'_i$. Then matching $\bar{M}$ consists of $n - d_r - p$ good pairs and $p$ bad pairs.
Therefore, using Theorem~\ref{thm:property}
\begin{eqnarray*}
&&n - (n-d_r-p) -p > (1-\eps)(n - (n-d_r-p)) 
\shortOnly{\\&}\Rightarrow\shortOnly{&}
 d_r > (1-\eps) (d_r+p) \Rightarrow p < \frac{\eps}{1-\eps}\cdot d_r
\end{eqnarray*}
This proves the following theorem:
\begin{theorem}
Any $\eps$-synchronization string along with the algorithm described in Section~\ref{sec:delOnly} form a linear-time streaming solution for deletion-only $(n, \delta)$-indexing problem guaranteeing $\frac{\eps}{1-\eps}\cdot n\delta$ misdecodings.
\end{theorem}

\subsection{Insertion Errors Only}\label{sec:insErrors}
We now depart to another simplified case where adversary is restricted to only insert symbols. We propose a decoding algorithm whose output is guaranteed to be error-free and contain less than $\frac{n\delta}{1-\eps}$ misdecodings.

Assume that $d_i$ symbols are inserted into the string $S$ to turn it in into $S'$ of size $n+d_i$ on the receiving side. Again, it is clear that there exists a monotone matching $M$ of size $n$ like 
$M = \{(1, t_1), (2, t_2), \cdots, ({n},t_{n})\}$
between $S$ and $S'$ that matches each symbol in $S$ to its actual index when it arrives at the receiver. 

The decoding algorithm we present, matches $S[i]$ to $S'[t'_i]$ in its output, $M'$, if and only if in all possible monotone matchings between $S$ and $S'$ that saturate $S$ (i.e., are of size $|S| = n$), $S[i]$ is matched to $S'[t'_i]$.  Note that any symbol $S[i]$ that is matched to $S'[t'_i]$ in $M'$ has to be matched to the same element in $M$; therefore, the output of this algorithm does not contain any incorrectly decoded indices; therefore, the algorithm is error-free.

Now, we are going to first provide a linear time approach to implement this algorithm and then prove an upper-bound of $\frac{d_i}{1-\eps}$ on the number of misdecodings. To this end, we make use of the following lemma:
\begin{lemma}\label{lem:leftright}
Let $M_L = \{(1, l_1), (2, l_2), \cdots, (n, l_n)\}$ be a monotone matching between $S$ and $S'$ such that $l_1, \cdots, l_{n}$ has the smallest possible value lexicographically. We call $M_L$ the \emph{left-most} matching between $S$ and $S'$. Similarly, let $M_R = \{(1, r_1), \cdots, (n, r_n)\}$ be the monotone matching such that $r_n, \cdots, r_{1}$ has the largest possible lexicographical value. Then $S[i]$ is matched to $S'[t'_i]$ in all possible monotone matchings of size $n$ between $S$ and $S'$ if and only if $(i, t'_i) \in M_R\cap M_L$.
\end{lemma}

This lemma can be proved by a simple contradiction argument.
Our algorithm starts by computing left-most and right-most monotone matchings between $S$ and $S'$ using the trivial greedy algorithm introduced in Section~\ref{sec:delOnly} on $(S, S')$ and them reversed. It then outputs the intersection of these two matching as the answer. This algorithm clearly runs in linear time.

To analyze this algorithm, we bound the number of successfully transmitted symbols that the algorithm refuses to decode, denoted by $p$. To bound the number of such indices, we make use of the fact that $n-p$ elements of $S'$ are matched to the same element of $S$ in both $M_L$ and $M_R$. As there are $p$ elements in $S$ that are matched to different elements in $S'$ and there is a total of $n+d_i$ elements in $S'$, there has to be at least $2p -[(n+d_i)-(n-p)] = p-d_i$ elements in $S'$ who are matched to different elements of $S$ in $M_L$ and $M_R$.

Consider the following monotone matching from $S$ to itself as follows:
\begin{eqnarray*}
\fullOnly{M& =& \{(i, i): \text{If $S[i]$ is matched to the same position of $S'$ in both $M$ and $M'$} \} \\}
\shortOnly{M& =& \{(i, i): \text{If $S[i]$ is matched to the same position of }
\\&&\text{$S'$ in both $M$ and $M'$} \} \\}
&&\cup\ \{(i, j): \exists k \text{ s.t. } (i,k)\in M_L, (j,k) \in M_R\}
\end{eqnarray*}
Note that monotonicity follows the fact that both $M_L$ and $M_R$ are both monotone matchings between $S$ and $S'$.
We have shown that the size of the second set is at least $p - d_i$ and the size of the first set is by definition $n-p$. Also, all pairs in the first set are good pairs and all in the second one are bad pairs. Therefore, by Theorem~\ref{thm:property}:
$$(n - (n-p) - (p-d_i)) > (1-\eps)(n-(n-p)) \Rightarrow p < \frac{d_i}{1-\eps}$$
which proves the efficiency claim and gives the following theorem.
\begin{theorem}\label{thm:errorFreeInsOnly}
Any $\eps$-synchronization string along with the algorithm described in Section~\ref{sec:insErrors} form a linear-time error-free solution for insertion-only $(n, \delta)$-indexing problem guaranteeing $\frac{1}{1-\eps}\cdot n\delta$ misdecodings.
\end{theorem}

Finally, we remark that a similar non-streaming algorithm can be applied to the case of deletion-only errors. Namely, one can compute the left-most and right-most matchings between the received string and string that is supposed to be received and output the common edges. By a similar argument as above, one can prove the following:
\begin{theorem}\label{thm:errorFreeDelOnly}
Any $\eps$-synchronization string along with the algorithm described in Section~\ref{sec:insErrors} form a linear-time error-free solution for deletion-only $(n, \delta)$-indexing problem guaranteeing $\frac{\eps}{1-\eps}\cdot n\delta$ misdecodings.
\end{theorem}

In the same manner as Theorem~\ref{thm:main}, we can derive the following theorem concerning deletion-only and insertion-only codes based on Theorems~\ref{thm:errorFreeInsOnly} and \ref{thm:errorFreeDelOnly}.

\begin{theorem}\label{thm:InsOnlyDelOnlyCodes}
For any $\eps>0$ and $\delta \in (0,1)$:
\begin{itemize}
\item There exists an encoding map $E: \Sigma^k \rightarrow \Sigma^n$ and a decoding map $D: \Sigma^* \rightarrow \Sigma^k$ such that if $x$ is a subsequence of $E(m)$ where $|x| \ge n-n\delta$ then $D(x) = m$. Further $\frac{k}{n} > 1 - \delta - \eps$, $|\Sigma|=f(\eps)$, and $E$ and $D$ are explicit and have linear running times in $n$. 
\item There exists an encoding map $E: \Sigma^k \rightarrow \Sigma^n$ and a decoding map $D: \Sigma^* \rightarrow \Sigma^k$ such that if $E(m)$ is a subsequence of $x$ where $|x| \le n+n\delta$ then $D(x) = m$. Further $\frac{k}{n} > 1 - \delta - \eps$, $|\Sigma|=f(\eps)$, and $E$ and $D$ are explicit and have linear running times in $n$. 
\end{itemize}
\end{theorem}

Finally, we remark that since indexing solutions offered in Theorems~\ref{thm:errorFreeInsOnly} and \ref{thm:errorFreeDelOnly} are error free, it suffices to use good erasure codes along with synchronization strings to obtain Theorem~\ref{thm:InsOnlyDelOnlyCodes}.

\global\def\RSPD{
\fullOnly{\subsection{Decoding Using the Relative Suffix Pseudo-Distance (RSPD)}\label{sec:improvedStreaming}}
\shortOnly{\section{Decoding Using the Relative Suffix Pseudo-Distance (RSPD)}\label{sec:improvedStreaming}}

In this section, we show how one can slightly improve the constants in the results obtained in Section~\ref{sec:sync_decoding} by replacing RSD with a related notion of ``distance'' between two strings introduced in \cite{braverman2015coding}. We call this notion \emph{relative suffix pseudo-distance} or RSPD both to distinguish it from our RSD relative suffix distance and also because RSPD is not a metric distance per se -- it is neither symmetric nor satisfies the triangle inequality.

\begin{definition}[Relative Suffix Pseudo-Distance (RSPD)]
Given any two strings $c,\tilde{c} \in \Sigma^*$, the \emph{suffix distance} between $c$ and $\tilde{c}$ is 
$$RSPD\left(c,\tilde{c}\right) = \min_{\tau: c \to \tilde{c}}\left\{\max_{i = 1}^{|\tau_1|} \left\{ \frac{sc\left(\tau_1\left[i,|\tau_1|\right]\right) + sc\left(\tau_2\left[i,|\tau_2|\right]\right)}{|\tau_1| - i + 1 - sc\left(\tau_1\left[i,|\tau_1|\right]\right)}\right\}\right\}$$
\end{definition}

We derive our algorithms by proving a useful property of synchronization strings: 

\begin{lemma}\label{lemma:SD_unique}
Let $S \in \Sigma^n$ be an $\eps$-synchronization string and $\tilde{c}\in\Sigma^m$. Then there exists at most one $c \in \bigcup_{i=1}^n S[1..i]$ such that $RSPD(c, \tilde{c}) \le 1-\eps$.
\end{lemma}

%\global\def\ProofofSDunique{
Before proceeding to the proof of Lemma~\ref{lemma:SD_unique}, we prove the following lemma:
\begin{lemma}\label{lemma:suffix}
Let $RSPD(S, T) \le 1-\eps$, then:
\begin{enumerate}
\item For every $1 \le s \le |S|$, there exists $t$ such that $ED\left(S[s,|S|], T\left[t, |T|\right]\right) \le (1-\eps)(|S| -  s + 1)$.
\item For every $1 \le t \le |T|$, there exists $s$ such that $ED\left(S[s,|S|], T\left[t, |T|\right]\right) \le (1-\eps)(|S| -  s + 1)$.
\end{enumerate}
\end{lemma}
\begin{proof}
\begin{figure}
\centering
\includegraphics[scale=.40]{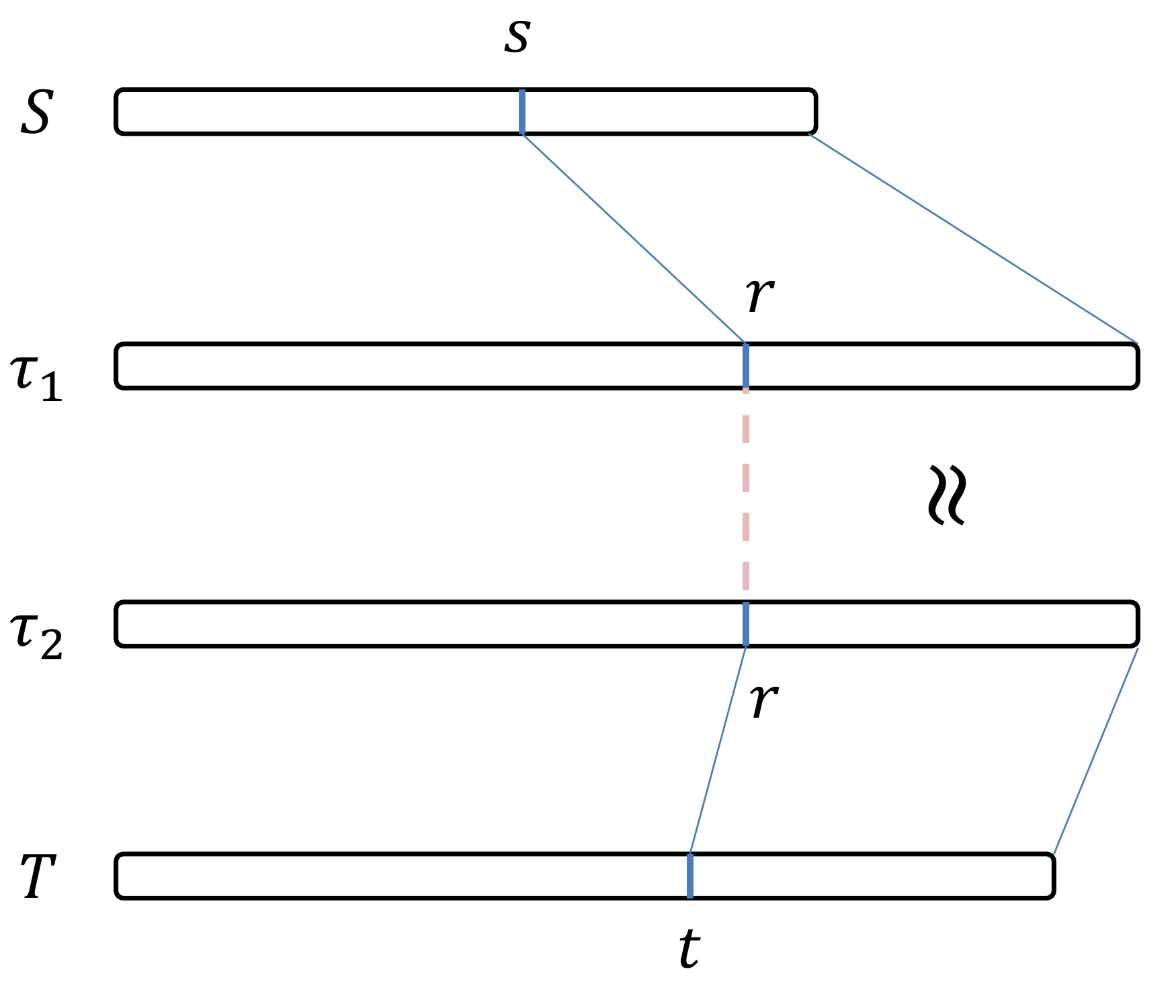}
\caption{Pictorial representation of the notation used in Lemma~\ref{lemma:suffix}}
\label{fig:suffix_distance}
\end{figure}
\-\paragraph{Part 1}
Let $\tau$ be the string matching chosen in $RSPD(S, T)$. There exist some $r$ such that $del(\tau_1\left[r, |\tau_1|\right]) = S[s, |S|]$. Note that $del(\tau_2[r, |\tau_2|])$ is a suffix of $T$. Therefore, there exists some $t$ such that $T[t, |T|] = del(\tau_2[r, |\tau_2|])$. Now,
\begin{eqnarray}
ED(S[s, |S|], T[t, |T|]) &\le& sc(del(\tau_1\left[r, |\tau_1|\right])) + sc(del(\tau_2\left[r, |\tau_1|\right]))\nonumber\\
&=& \frac{sc(del(\tau_1\left[r, |\tau_1|\right])) + sc(del(\tau_2\left[r, |\tau_1|\right]))}{|\tau_1| - r + 1 - sc(\tau_1[r, |\tau_1|])} \cdot (|\tau_1| - r + 1 - sc(\tau_1[r, |\tau_1|]))\nonumber\\
&\le& RSPD(S, T) \cdot (|S| - s + 1)\nonumber\\
&\le& (1-\eps) \cdot (|S| - s + 1)\label{eq:SD-eq3}
\end{eqnarray}
\-\paragraph{Part 2}
Similarly, let $\tau$ be the string matching chosen in $RSPD(S, T)$. There exists some $r$ such that $del(\tau_2\left[r, |\tau_2|\right]) = T[t, |T|]$. Now, $del(\tau_1[r, |\tau_1|])$ is a suffix of $S$. Therefore, there exists some $s$ such that $S[s, |S|] = del(\tau_1[r, |\tau_1|])$. Now, all the steps we took to prove equation \eqref{eq:SD-eq3} hold and the proof is complete.
\end{proof}

\begin{algorithm}
\caption{Synchronization string decode}
\begin{algorithmic}[1]\label{alg:SyncDecode}
\REQUIRE A received message $\tilde{c} \in \Sigma^m$ and an $\eps$-synchronization string $S \in \Sigma^n$

\STATE $ ans \leftarrow \emptyset$
\FOR{Any prefix $c$ of $S$}
\STATE $d[i][j][l] \leftarrow \min_{\substack{\tau:c(i)\rightarrow \tilde{c}(j)\\sc\left(\tau_1\right) = l}} \max_{k=1}^{|\tau_1|} 
\frac{sc\left(\tau_1\left[k..\left|\tau_1\right|\right]\right)+sc\left(\tau_2\left[k..\left|\tau_2\right|\right]\right)}
{|\tau_1|-k+1+sc\left(\tau_1\left[k..\left|\tau_1\right|\right]\right)}$
\medskip
\STATE $RSPD(c, \tilde{c})\leftarrow \min_{l' = 0}^{|\tilde{c}|}d[\texttt{i}][|\tilde{c}|][l']$ \label{step:minimize}
\IF {$RSPD(c, \tilde{c}) \leq 1-\eps$}
\STATE $ans \leftarrow c$
\ENDIF
\ENDFOR
\ENSURE $ans$
\end{algorithmic}
\end{algorithm}

\fullOnly{\begin{proof}[Proof of Lemma~\ref{lemma:SD_unique}]}
\shortOnly{\begin{proof}[of Lemma~\ref{lemma:SD_unique}]}
For a contradiction, suppose that there exist a $\tilde{c}$, $l$ and $l'$ such that $l < l'$ and $RSPD(S[1, l], \tilde{c}) \leq 1-\eps$ and $RSPD(S[1, l'], \tilde{c}) \leq 1-\eps$. Now, using part 1 of Lemma~\ref{lemma:suffix}, there exists $k$ such that $ED\left(S[l+1, l'], \tilde{c}[k, |\tilde{c}|]\right) \le (1-\eps) (l' - l)$. Further, part 2 of Lemma~\ref{lemma:suffix} gives that there exist $l''$ such that $ED\left(S[l''+1, l], \tilde{c}[k, |\tilde{c}|]\right) \le (1-\eps)(l - l'')$. Hence, $$ED(S[l+1, l'], S[l'+1, l'']) \le ED(S[l+1, l'], \tilde{c}[k, |\tilde{c}|]) + ED(S[l'+1, l''], \tilde{c}[k, |\tilde{c}|]) \le (1-\eps) (l' - l'')$$ which contradicts the fact that $S$ is an $\eps$-synchronization string.
\end{proof}
%}

%\fullOnly{\ProofofSDunique}
%\ProofofSDunique

Lemma~\ref{lemma:SD_unique} implies a natural algorithm for decoding $\tilde{c}$: simply search over all prefixes of $S$ for the one with small enough suffix distance from $\tilde{c}$. We prove that this is possible in $O(n^5)$ via dynamic programming. \shortOnly{See Algorithm~\ref{alg:SyncDecode} for the pseudocode.}

\begin{theorem}\label{thm:EfficientDecoding}
Let $S\in\Sigma^n$ be an $\eps$-synchronization string, and $\tilde{c}\in \Sigma^m$. Then Algorithm~\ref{alg:SyncDecode}, given input $S$ and $\tilde{c}$, either returns the unique prefix $c$ of $S$ such that $RSPD(c, \tilde{c})\le 1-\eps$ or returns $\emptyset$ if no such prefix exists. Moreover, Algorithm~\ref{alg:SyncDecode} runs in time $O(n^5)$; spending $O(n^4)$ for each received symbol.
\end{theorem}

%\global\def\ProofofThmEfficientDecoding{
%\shortOnly{\begin{proof}[Proof of Theorem~\ref{thm:EfficientDecoding}]}
%\fullOnly{\begin{proof}}
\begin{proof}
To find $c$, we calculate the RSPD of $\tilde{c}$ and all prefixes of $S$ one by one. We only need to show that the RSPD of two strings of length at most $n$ can be found in $O(n^3)$. We do this using dynamic programming. Let us try to find $RSPD(s, t)$. Further, let $s(i)$ represent the suffix of $s$ of length $i$ and $t(j)$ represent the suffix of $t$ of length $j$. Now, let $d[i][j][l]$ be the minimum string matching $(\tau_1,\tau_2)$ from $s(i)$ to $t(j)$ such that $sc(\tau_1) = l$. In other words,
$$d[i][j][l] = \min_{\substack{\tau:s(i)\rightarrow t(j)\\sc\left(\tau_1\right) = l}} \max_{k=1}^{|\tau_1|} 
\frac{sc\left(\tau_1\left[k..\left|\tau_1\right|\right]\right)+sc\left(\tau_2\left[k..\left|\tau_2\right|\right]\right)}
{|\tau_1|-k+1+sc\left(\tau_1\left[k..\left|\tau_1\right|\right]\right)},$$ where $\tau$ is a string matching for $s(i)$ and $t(j)$. Note that for any $\tau:s(i)\rightarrow t(j)$, one the following three scenarios might happen:
\begin{enumerate}
\item $\tau_1(1) = \tau_2(1) = s\left(|s| - (i - 1)\right) = t(|t| - (j - 1))$: In this case, removing the first elements of $\tau_1$ and $\tau_2$ gives a valid string matching from $s(i-1)$ to $t(j-1)$.
\item $\tau_1(1) = * \text{ and } \tau_2(1) = t(|t| - (j - 1))$: In this case, removing the first element of $\tau_1$ and $\tau_2$ gives a valid string matching from $s(i)$ to $t(j-1)$.
\item $\tau_2(1) = * \text{ and } \tau_1(1) = s(|s| - (i - 1))$: In this case, removing the first element of $\tau_1$ and $\tau_2$ gives a valid string matching from $s(i-1)$ to $t(j)$.
\end{enumerate}

This implies that
\begin{eqnarray*}
d[i][j][l] = \min\Bigg\{&& 
d[i-1][j-1][l] \text{ (Only if }s(i) =t(j)\text{)},\\
&&\max\left\{d[i][j-1][l-1], \frac{l+(j-(i-l))}{(i+l)+l}\right\},\\
&&\max\left\{d[i-1][j][l], \frac{l+(j-(i-l))}{(i+l)+l}\right\}\Bigg\}.
\end{eqnarray*}
Hence, one can find $RSPD(s,t)$ by minimizing $d[|s|][|t|][l]$ over all possible values of $l$, as Algorithm~\ref{alg:SyncDecode} does in Step~\ref{step:minimize} for all prefixes of $S$. Finally, Algorithm~\ref{alg:SyncDecode} returns the prefix $c$ such that $RSPD(s,t) \leq \frac{1-\eps}{2}$ if one exists, and otherwise it returns $\emptyset$.
\end{proof}
%}

%\fullOnly{\ProofofThmEfficientDecoding}

We conclude by showing that if an $\eps$-synchronization string of length $n$ is used along with the minimum RSPD algorithm, the number of misdecodings will be at most $\frac{n\delta}{1-\eps}$.

\begin{theorem}\label{thm:DecodingGuarantee}
Suppose that $S$ is an $\eps$-synchronization string of length $n$ over alphabet $\Sigma$ that is sent over an insertion-deletion channel with $c_i$ insertions and $c_d$ deletions.
By using Algorithm~\ref{alg:SyncDecode} for decoding the indices, the outcome will contain less than $\frac{c_i}{1-\eps} + \frac{c_d\eps}{1-\eps}$ misdecodings.
\end{theorem}

%\global\def\ProofofThmDecodingGuarantee{
%\shortOnly{\begin{proof}[Proof of Theorem~\ref{thm:DecodingGuarantee}]}
%\fullOnly{\begin{proof}}
\begin{proof}
The proof of this theorem is similar to the proof of Theorem~\ref{thm:MinRSDMisDec}. Let prefix $S[1, i]$ be sent through the channel $S_\tau[1, j]$ be received on the other end as the result of adversary's set of actions $\tau$. Further, assume that $S_\tau[j]$ is successfully transmitted and is actually $S[i]$ sent be the other end. We first show that $RSPD(S[1, i], S'[1, j])$ is less than the relative suffix error density:
\begin{eqnarray*}
RSPD(S[1, i], S'[1, j]) & = & \min_{\tilde\tau: c \to \tilde{c}}\left\{\max_{k = 1}^{|\tilde\tau_1|} \left\{ \frac{sc\left(\tilde\tau_1\left[k,|\tilde\tau_1|\right]\right) + sc\left(\tilde\tau_2\left[k,|\tilde\tau_2|\right]\right)}{|\tilde\tau_1| - k + 1 - sc\left(\tilde\tau_1\left[k,|\tilde\tau_1|\right]\right)}\right\}\right\}\\
&\le& \max_{k = 1}^{|\tau_1|} \left\{ \frac{sc\left(\tau_1\left[k,|\tau_1|\right]\right) + sc\left(\tau_2\left[k,|\tau_2|\right]\right)}{|\tau_1| - k + 1 - sc\left(\tau_1\left[k,|\tau_1|\right]\right)}\right\}\\
&=& \max_{j\le i}\frac{\mathcal{E}(j, i)}{i - j} = \text{ Relative Suffix Error Density}
\end{eqnarray*}

Now, using Theorem~\ref{lem:NumberofBadSuffixDistances}, we know that the relative suffix error density is smaller than $1-\eps$ upon arrival of all but at most $\frac{c_i+d_d}{1-\eps}-c_d$ of successfully transmitted symbols. 
Along with Lemma~\ref{lemma:SD_unique}, this results into the conclusion that the minimum RSPD decoding guarantees $\frac{c_i}{1-\eps} + c_d \left(\frac{1}{1-\eps}-1\right)$ misdecodings.
This finishes the proof of the theorem.
\end{proof}
%}
%\fullOnly{\ProofofThmDecodingGuarantee}
}
\fullOnly{\RSPD}

}

\section*{Acknowledgements}
The authors thank Ellen Vitercik and Allison Bishop for valuable discussions in the early stages of this work. 
\newpage
\shortOnly{
\begin{center}
\bfseries \huge Appendices
\end{center}

\appendix

\shortOnly{
%\section{Introduction}\label{app:Introduction}
%\RelatedWork
%\section{Preliminaries}\label{app:ECCPrelims}
%\ECCPrelims
%\section{More on the Indexing Problem}\label{app:MoreOnIndexing}
%\FurtherOnIndexing
\section{Proof of Theorems in Section~\ref{sec:codings}}\label{app:ECCviaIndexing}
\ProofOfMisDecodeToHalfError
\EncDecAlgorithms
\section{Proof of Theorems in Section~\ref{sec:sync}}\label{app:chapter3}

\ProofOfMetricPropertiesOfRSD

We now move on to proving Theorem~\ref{thm:existence}. Its prove requires the general Lov\'{a}sz local lemma which we recall first here:

\StatementofLemLLL

\ProofofThmExistence

\AlphabetSizeRemarks

\ProofOfInfiniteSynchExistence

\subsection{Decoding}\label{app:MinRSDDecoding}

\RSDMinDistanceDecoding

\RSPD

\section{Proof of Theorems in Section~\ref{sec:improved_decoding}}\label{app:improved_decoding}
\SynchToSelfMatchingProperty

\ProofOfEpsBadIndsInEpsBadStrings

\ProofOfEpsSelfMatchingRandomConstruction

\ProofOfLimitedIndependenceRelaxedConstruction

\exclude{
\ProofofThmFiniteSyncConstruction

\ProofofThminfiniteSync

}

}

}

%\newpage
\bibliographystyle{plain}
\bibliography{bibliography}

\end{document}